\documentclass[11pt]{article}
\usepackage{authblk}
\usepackage[a4paper,margin=2cm]{geometry}

\usepackage[utf8]{inputenc}
\usepackage{amsmath}
\usepackage{amssymb}
\usepackage{amsthm}
\usepackage{dsfont}
\usepackage{color}
\usepackage{amsmath}
\usepackage{amsfonts}
\usepackage{amssymb}
\usepackage{bbm}
\usepackage{lineno}
\usepackage[colorlinks=true,linkcolor=red,citecolor=blue]{hyperref}
\usepackage{relsize}
\usepackage{pst-node}
\usepackage[procnumbered, linesnumbered,algoruled]{algorithm2e}
\usepackage[compact]{titlesec}
\usepackage{tikz}
\usetikzlibrary{positioning}
\usetikzlibrary {shapes,matrix}
\usetikzlibrary{arrows}
\newcommand{\omitmac}[1]{}
\tikzset{
	semithick,
	node distance = 2cm,
	dot/.style={circle,fill,inner sep=2pt}
}
\tikzset{
	side by side/.style 2 args={
		line width=2pt,
		#1,
		postaction={
			clip,postaction={draw,#2}
		}
	}
}
\tikzstyle{every state}=[draw = black,thick,fill = white,minimum size = 4mm]
\tikzstyle{selected edge} = [draw,line width=2pt,-,red!50]
\tikzset{
	edge/.style={->,> = latex'}
}
\usepackage{hyperref}
\newcommand{\abs}[1]{{\left|#1 \right|}}
\newcommand{\bm}{\bar{m}}
\newcommand{\bq}{{\bar{q}}}
\newcommand{\cM}{\mathcal{M}}
\newcommand{\ck}{\mathcal{K}}
\newcommand{\negA}{\vspace{-0.05in}}
\newcommand{\negB}{\vspace{-0.1in}}
\newcommand{\negC}{\vspace{-0.18in}}
\newcommand{\posA}{\vspace{0.08in}}
\newcommand{\mysection}[1]{\negC\section{#1}\negA}
\newcommand{\mysubsection}[1]{\negB\subsection{#1}\negA}
\newcommand{\mysubsubsection}[1]{\negB\subsubsection{#1}\negA}
\newcommand{\myparagraph}[1]{\par\smallskip\par\noindent{\bf{}#1:~}}
\newcommand{\ceil}[1]{{\left\lceil#1  \right\rceil}}
\newcommand{\comment}[1]{}
\usepackage[nameinlink]{cleveref}
\usepackage[procnumbered, linesnumbered,algoruled]{algorithm2e}
\newcommand{\cm}{{\mathcal{M}}}
\newcommand{\cA}{{\mathcal{A}}}
\newcommand{\cB}{{\mathcal{B}}}
\newcommand{\entropy}{{\mathcal{H}}}

\def \NN   {{\cal N}}
\def \CC   {{\cal C}}
\newcommand{\cI}{{\mathcal{I}}}
\newcommand{\cF}{{\mathcal{F}}}
\newcommand{\cC}{{\mathcal{C}}}
\newcommand{\cS}{{\mathcal{S}}}
\newcommand{\bx}{{\bar{x}}}
\newcommand{\by}{{\bar{y}}}
\newcommand{\bz}{{\bar{z}}}
\newcommand{\bc}{{\bar{c}}}
\newcommand{\OPT}{\textnormal{OPT}}
\newcommand{\APX}{\textnormal{APX}}
\newcommand{\opt}{\textnormal{opt}}
\newcommand{\apx}{\textnormal{apx}}
\newcommand{\topt}{\tilde{\textnormal{opt}}}
\newcommand{\greedy}{\textsc{Greedy}}
\newcommand{\greedyplus}{\textsc{Greedy+Singleton}}
\newcommand{\enum}{\textsc{EnumGreedy}}
\newcommand{\eps}{{\varepsilon}}
\newcommand{\aeps}{{T}}
\newcommand{\E}{{\mathbb{E}}}
\newcommand{\pP}{\textnormal{\textsf{p-PTAS}}}
\newcommand{\pE}{\textnormal{\textsf{p-EPTAS}}}
\newcommand{\pF}{\textnormal{\textsf{p-FPTAS}}}
\newcommand{\PTAS}{\textnormal{\textsf{PTAS}}}
\newcommand{\EPTAS}{\textnormal{\textsf{EPTAS}}}
\newcommand{\FPTAS}{\textnormal{\textsf{FPTAS}}}
\newcommand{\floor}[1]{\left\lfloor #1 \right\rfloor}
\DeclareMathOperator*{\argmax}{arg\,max}
\DeclareMathOperator*{\argmin}{arg\,min}
\newcommand{\klf}{{\varphi}}
\crefname{claim}{claim}{claims}
\crefname{obs}{observation}{observations}
\newcommand{\dynamic}{\textnormal{\textsf{Dynamic}}}
\newcommand{\MVK}{\textnormal{\textsf{MVK}}}
\newcommand{\UFP}{\textnormal{\textsf{UFP}}}
\newcommand{\BagUFP}{\textnormal{\textsf{BagUFP}}}
\newcommand{\Sol}{\textnormal{\texttt{DP}}}

\newtheorem{thm}{Theorem}[section]

\newtheorem{obs}[thm]{Observation}

\newtheorem{cor}[thm]{Corollary}
\newtheorem{lemma}[thm]{Lemma}
\newtheorem{definition}[thm]{Definition}
\newtheorem{theorem}[thm]{Theorem}

\newtheorem{claim}[thm]{Claim}
\Crefname{obs}{Observation}{Observations}

\usepackage[colorinlistoftodos,textsize=tiny,textwidth=2cm,color=green!50!gray]{todonotes}
\usepackage{geometry}

\def\DEBUG{true}

\ifdefined\DEBUG
\newcommand{\fab}[1]{\textcolor{blue}{#1}}
\usepackage{xcolor}
\definecolor{darkgreen}{rgb}{0.9, 0.2, 0.1}
\newcommand{\ilan}[1]{\textcolor{darkgreen}{#1}}
\newcommand{\ariel}[1]{\textcolor{red}{#1}}
\def\rem#1{\marginpar{\raggedright\scriptsize #1}}
\newcommand{\fabr}[1]{\todo[color=cyan!100!black!50]{F: #1}}

\else
\newcommand{\fab}[1]{#1}
\newcommand{\fabr}[1]{}
\fi

\begin{document}

\def \II   {I}
\newcommand{\one}{\mathbbm{1}}
	\def\claimproof{\proof}
\def\endclaimproof{\hfill$\square$\\}

\newcommand{\guess}{\tilde{\opt}}

\renewcommand\qedsymbol{$\blacksquare$}

\title{
Unsplittable Flow on a Short Path
}

\author{Ilan Doron-Arad  \and Fabrizio Grandoni \and Ariel Kulik}
\date{}
\maketitle

\begin{abstract}
\noindent In the Unsplittable Flow on a Path problem (\UFP), we are given a path graph with edge capacities and a collection of tasks. Each task is characterized by a demand, a profit, and a subpath. Our goal is to select a maximum profit subset of tasks such that the total demand of the selected tasks that use each edge $e$ is at most the capacity of $e$. \BagUFP\ is the generalization of \UFP\ where tasks are partitioned into bags, and we are allowed to select at most one task per bag. \UFP\ admits a \PTAS\ [Grandoni,M{\"o}mke,Wiese'22] but not an \EPTAS\ [Wiese'17]. \BagUFP\ is APX-hard [Spieksma'99] and the current best approximation is $O(\log n/\log\log n)$ [Grandoni,Ingala,Uniyal'15], where $n$ is the number of tasks. 

In this paper, we study the mentioned two problems when parameterized by the number $m$ of edges in the graph, with the goal of designing faster parameterized approximation algorithms. We present a parameterized \EPTAS\ for \BagUFP, and a substantially faster parameterized \EPTAS\ for \UFP\ (which is an \FPTAS\ for $m=O(1)$). We also show that a parameterized \FPTAS\ for \UFP\ (hence for \BagUFP) does not exist, therefore our results are qualitatively tight.  
\end{abstract}


\section{Introduction}

In the classical Usplittable Flow on a Path problem (UFP) we are given an $m$-edge path graph $G=(V,E)$ with (non-negative integer) edge capacities $u:E\rightarrow \mathbb{N}$, and a collection of $n$ tasks $T$. Each task $i$ is characterized by a demand $d(i)\in \mathbb{N}$, a weight (or profit) $w(i)\in \mathbb{N}$, and a subpath $P(i)$\footnote{Throughtout this paper, for a subpath $P$, we sometimes use $P$ also to denote the corresponding set of edges $E(P)$: the meaning will be clear from the context.}. A feasible solution consists of a subset of (selected) tasks $S\subseteq T$ such that, for each edge $e$, $\sum_{i\in S: e\in P(i)}d(i)\leq u(e)$. In other words, the total demand of the selected tasks using each edge $e$ cannot exceed the capacity of $e$. Our goal is to compute a feasible solution $\OPT$ of maximum total profit $\opt=w(\OPT):=\sum_{i\in \OPT}w(i)$.

UPF has several direct and undirect applications \cite{BBFNS00,BBCR06,BSW14,CCKR11,CMS07,CHT02,CWMX10,DPS10,LMV00,PUW00}. For example, one might interpret $G$ as a time interval subdivided into time slots (the edges). At each time slot we are given some amount of a considered resource, say, energy. The tasks represent jobs that we might execute, therefore gaining a profit. However each executed job will consume some amount of the shared resource during its execution, thus we might not be able to execute all the jobs (hence we need to perform a selection). 

UFP is strongly NP-hard \cite{BSW14,CWMX10} and it is well-studied in terms of approximation algorithms. After a long sequence of improvements \cite{AGLW13,AGLW14,BCES06,BFKS09,BGKMW15,BSW14,CCGK07,CEK09,GMW21,GMWZ17,GMWZ18,GMW22}, a \PTAS\ for \UFP\ was eventually achieved by Grandoni, M{\"o}mke and Wiese \cite{GMW22STOC}. We recall that a \PTAS\ (for a maximization problem) is an algorithm parameterized by $\eps>0$, which provides a $(1-\eps)$ approximation in time $|I|^{O_\eps(1)}$, where $|I|$ is the input size. \EPTAS s and \FPTAS s are defined similarly, however with running times of the form $f(1/\eps)\cdot |I|^{O(1)}$ and $(|I|/\eps)^{O(1)}$, resp., where $f(\cdot)$ is a computable function. Wiese \cite{W17} proved that \UFP, parameterized by the number of selected tasks, is $W[1]$-hard: this excludes the existence of an \EPTAS\ for \UFP\ by standard reductions (unless $\textnormal{FPT=W[1]}$ \cite{M08}).

In the above scenario there is no flexibility on the time when a job is executed. The \BagUFP\ problem is a generalization of \UFP\ which was introduced to allow for such flexibility. Here we are given the same input as \UFP, plus a partition of the tasks into $\ell$ bags $\cB = \{B_1,\ldots, B_{\ell}\}$, $\dot\cup_{j=1}^{\ell}B_j=T$. A feasible solution $S$ has to satisfy the capacity constraints as in \UFP, plus the extra constraint that at most one task per bag can be selected, namely $|S\cap B_j|\leq 1$ for $j=1,\ldots,\ell$. This easily captures jobs that can be executed at different times (and even more general settings). For example, if a job can be executed within a given time window (also known as the time-windows \UFP\ problem), it is sufficient to create a bag that contains multiple copies of the same task which differ only in the subpath $P(i)$ (with one subpath per potential valid scheduling time). \BagUFP\ is APX-hard \cite{S99}, which rules out the existence of a \PTAS\ for it. The current best approximation ratio for \BagUFP\ is $O(\log n/\log\log n)$ \cite{GIU15}, slightly improving on the $O(\log n)$ approximation in \cite{CCGRS14}. A constant approximation for \BagUFP\ is known for the cardinality version of the problem \cite{GIU15}, i.e. when all the profits are $1$.

\subsection{Our Results and Techniques}

The mentioned \PTAS\ for \UFP\ \cite{GMW22STOC} has a very poor dependence on $\eps$ in the running time, which makes it most likely impractical. Though an improvement of the running time is certainly possible, as mentioned before an \EPTAS\ for \UFP\ does not exist (unless $\textnormal{FPT}=\textnormal{W}[1]$). The situation for \BagUFP\ is even worse: here even a \PTAS\ does not exist (unless $\textnormal{P}=\textnormal{NP}$), and currently finding a constant approximation algorithm (which might exist) is a challenging open problem. 

Motivated by the above situation, it makes sense to consider parameterized approximation algorithms for \UFP\ and \BagUFP. The general goal here is to identify some integer parameter $p$ that captures some relevant aspect of the input (or some property of the output), and try to design approximation algorithms whose running time is better than the state of the art when $p$ is sufficiently small. In particular a parameterized \PTAS\ (\pP) is defined similarly to a \PTAS, however with running time of the form $f(p)|\II|^{O_\eps(1)}$ for some commutable function $f(\cdot)$. Parameterized \EPTAS\ (\pE) and parameterized \FPTAS\ (\pF) are defined similarly, w.r.t. \EPTAS\ and \FPTAS\ resp. More explicitly, a \pE\ has a running time of the form $f\left(p+1/\eps\right)|\II|^{O(1)}$, while a \pF\ has a running time of the form $f(p)(|\II|/\eps)^{O(1)}$. For a meaningful choice of $p$, it makes sense to search for a \pE\ (or better) for \UFP, and for a \pP\ (or better) for \BagUFP.  

Probably the most standard parameter is the number $k$ of tasks in the desired solution. This is also the objective function for the cardinality version of the problems (with profits equal to $1$). Wiese \cite{W17} proved that \UFP\ is $W[1]$-hard under this parametrization, which rules out a \pE. He also presented a \pP\ for the cardinality version of \UFP\ with parameter $k$ (later improved by the \PTAS\ in \cite{GMW22STOC}, which also works for arbitrary profits). To the best of our knowledge, the same parametrization of \BagUFP\ was not studied in the literature.

In this paper we focus on {the parameter $m$, namely the number of edges in $G$ - the length of the path}. This makes sense in the realistic scenarios where $n\gg m$ { i.e., 
there are significantly more jobs than time slots. For example, such \UFP\ instances occur in personnel scheduling \cite{van2013personnel,baker1976workforce,bechtold1991comparative,alfares2004survey} where, e.g., workers are assigned to shifts within a working day ($m \approx 8$ working hours), or for an interval of days in the week ($m = 7$ days).} We achieve the following main results:

\subsubsection{Algorithms and Hardness for \BagUFP} 

A simple reduction from Partition shows that (assuming $\textnormal{P} \neq\textnormal{NP}$) there is no \FPTAS\ for \BagUFP\ even for $m=2$ (for $m=1$ an \FPTAS\ exists since the problem is equivalent to Multiple Choice Knapsack). As an obvious corollary, there is no \pF\ with parameter $m$ for the same problem (see  \Cref{sec:BagUFPhardness}).
\begin{theorem}
		\label{thm:BagUFPnoFPTAS}
	Unless $\textnormal{P} = \textnormal{NP}$, there is no \FPTAS\ for \BagUFP\ even in the case $m = 2$. 
\end{theorem}
\begin{cor}
	\label{cor:BagUFPnopFPTAS}
	Unless $\textnormal{P} = \textnormal{NP}$,  there is no \pF\ for \BagUFP\ parametrized by 
 the path length $m$.  
 \end{cor}

Hence, qualitatively speaking, the best one can hope for is a \pE. This is precisely what we achieve (see \Cref{sec:bag}).  
\begin{theorem}\label{thm:BagUFP-EPTAS}
There is a \pE\ for \BagUFP\ parametrized by the path length $m$. Its running time is $2^{\left(m / \eps^{1/\eps} \right)^{O(1)}} \cdot |I|^{O(1)}$.
\end{theorem}

Our approach substantially differs from previous algorithmic approaches for \UFP\ (see, e.g., \cite{GMW22STOC} and references therein) which relied on concepts such as classification of items by {\em demands} and probabilistic arguments. 
We observe that the bag constraints induce a matroid (more specifically, a partition matroid with capacity $1$ for each set). Therefore we consider the standard LP relaxation for a partition matroid (which has integral basic solutions), and augment it with the $m$ linear constraints corresponding to the capacity constraints. As proved in \cite{grandoni2010approximation}, a basic optimal solution $x^*$ to this LP (which can be computed in polynomial time for arbitrary $m$) has at most $2m$ fractional values (with value strictly between $0$ and $1$). The variables with value $1$ in $x^*$ induce a feasible \BagUFP\ solution with profit at least the optimal LP profit minus (almost) the profit of $2m$ tasks: this is problematic if the latter tasks have a profit comparable to $\opt=w(\OPT)$, where $\OPT$ is some reference optimal solution. 

We can avoid the above issue as follows. Let $H$ be the (\emph{heavy}) tasks with profit at least $\frac{\eps}{m}\opt$. We can guess the heavy tasks $H\cap \OPT$ in $\OPT$ (which are at most $m/\eps$ many), reduce the problem (i.e., remove all the tasks in the bags containing tasks from $H\cap \OPT$, remove tasks in $H$, and  reduce the available capacity of every edge by the total demand of $\OPT\cap H$ for the specific edge), and apply the mentioned LP-rounding technique to the remaining (\emph{light}) tasks. Now the drop of the fractional variables reduces the profit by at most $2\eps\cdot \opt$, leading to a $1+O(\eps)$ approximation. Unfortunately, this algorithm would take time $|H|^{\Omega(m/\eps)}$, which is still not compatible with a \pE. 

In order to circumvent the latter issue, we exploit the notion of representative sets, which was introduced in \cite{doron2023IPEC,doron2023ICALP,DKS23} to deal with a class of maximization problems with a single budget constraint. In contrast, we construct a representative set in the more general regime of multiple budget constraints imposed by the unsplittable flow setting. In more detail, in \pE\ time, we are able to compute a (representative) subset of tasks $R$ of size depending only on $m$ and $1/\eps$, such that there exists a nearly optimal solution $S$ such that $S\cap H \subseteq R$. Therefore, one can restrict to $R$ the above guessing of heavy tasks, which takes  $|R|^{O(m/\eps)}$ time: this is now compatible with a \pE. We remark that our techniques, combined with the representative set techniques of \cite{doron2023IPEC,doron2023ICALP,DKS23}, can give a \pE~for the more general problem of \UFP~with a general matroid constraint. We leave such efforts to the journal version of the paper. On the other hand, \UFP~with a general matroid is somewhat harder since an FPTAS is ruled out even for an instance with path of length $1$ (a single budget constraint) \cite{doron2023tight}. 


\subsubsection{Algorithms and Hardness for \UFP}

We start by showing that there is no \pF\ for \UFP\ parameterized by $m$. This, together with \Cref{thm:BagUFP-EPTAS}, gives a tight bound for UFP in the short path point-of-view~(see \Cref{sec:UFPhardness}). Notice that this is not implied by \Cref{thm:BagUFPnoFPTAS} since \UFP\ is a special case of \BagUFP.
\begin{theorem}
	\label{thm:hardnessUFP} 
	Unless \textnormal{FPT=W[1]}, there is no \pF~for \UFP\ parametrized by the path length $m$.  
\end{theorem}

Unlike previous hardness results \cite{BSW14,CWMX10,S99,W17,chrobak2012caching} for \UFP\ and its variant, which rely on a path of polynomial length in the input size, our lower bound requires having \UFP\ instances with a {\em short} path. Namely, the number of tasks is significantly larger than the length of the path. Our starting point is to obtain a hardness result for a {\em multiple choice} variant of $k$-subset sum in which the numbers are partitioned into sets $A_1,\ldots, A_k$, each set with $n$ numbers, and the goal is to select one number from each set such their sum  is  exactly  a given target value.  We use color-coding to show that multiple-choice $k$-subset sum does not have an FPT-algorithm unless W[1]=FPT (which may be useful for other hardness results). Then, we reduce multiple-choice $k$-subset sum to \UFP\ by constructing a \UFP\ instance with $m = O(k)$ edges and with polynomial weights. Roughly, we interpret the edges of the path in correspondence to the $k$ sets $A_1,\ldots, A_k$. The constructed instance has a pair of tasks $z^i_j,q^i_j$, with complementary subpaths, for every number $j = 1,\ldots,n$ in the $i$-th set $A_i$. Along with a carefully defined demand and weight functions, This \UFP\ instance satisfies that exactly $k$ pairs can be chosen for a sufficiently high weight if and only if the original subset sum instance has a solution. We remark that this construction utilizes the short path in 
a non-trivial manner. 


\Cref{thm:BagUFP-EPTAS} already provides a \pE\ for \UFP. We are however able to derive a \pE\ with a substantially better running time (see \Cref{sec:UFPalg}). 
\begin{theorem}
	\label{thm:eptasUFP}
	There is an \pE\ for \UFP\ parameterized by the path length $m$, with running time~$O\left(\frac{n^3}{\eps}+\left(\frac{1}{\eps}\right)^{O(m^2)}m^3\log n\right)$.
\end{theorem}
In particular, 
for 
$m\leq C  \cdot \sqrt{ \log_{\frac{1}{\eps}} n } $, for a sufficiently small constant $C>0$, 
 our running time 
 is the running time of an \FPTAS. 
We recall that achieving an \FPTAS\ (or even an \EPTAS) for \UFP\ in general is not possible and the previous state of the art for \UFP\ with a constant number of edges is the PTAS for the general problem \cite{GMW22STOC}.



The basic idea of the algorithm is a follows. Consider all the tasks $T_{\varphi}$ whose path is $\varphi$. Let $\opt_\varphi$ be the profit of some optimal solution $\OPT$ restricted to $T_\varphi$, i.e. $\opt_{\varphi}=w(\OPT\cap T_{\varphi})$. Given the value of $\opt_{\varphi}$, it is sufficient to find a minimum-demand subset of tasks $S_\varphi\subseteq T_{\varphi}$ with profit at least $\opt_{\varphi}$: the union of the sets $S_{\varphi}$ would be feasible and optimal. To achieve the target running time we use
this basic idea along with rounding of the weights and a coarse
guessing of the the values $\opt_{\varphi}$. By a standard rounding argument, we can assume that the weights are in $[\frac{n}{\eps}]$ while loosing a factor $1-\eps$ in the approximation.  
This allows us  to pre-compute the minimal demand subset of $T_{\varphi}$ which attains a threshold rounded weight, for every possible  threshold, using a standard dynamic program. 
The pre-computed subsets are used to reconstruct a solution $S_{\varphi}$ from the value of $\opt_{\varphi}$. 
Finally, we guess the values of $\opt_{\varphi}$ up to an additive error of $\approx \frac{\eps}{m^2} \cdot w(\OPT)$. This coarse guess of the values of $\opt_{\varphi}$ allows us to enumerate over all possible guesses within the 
running time, while only introducing an additional $1-\eps$ factor in the approximation.



\subsection{Preliminaries}
\label{sec:preliminaries}

For every $n \in \mathbb{N}$  we use $[n] = \{1,\ldots,n\}$. We use $(G,u,T,P,d,w,{\cal B})$ to denote a \BagUFP\ instance and by $(G,u,T,P,d,w)$ to denote a \UFP\ instance. 
Given and instance $I$ of \UFP\ or \BagUFP, we let $\OPT(I)$ denote some reference optimal solution, and $\opt(I)=w(\OPT(I))$ be its profit. We use $|I|$ to denote the encoding size of $I$. When $I$ is clear from the context, we simply use $\OPT$ and $\opt$, resp. 
Given a subset of tasks $S\subseteq T$, we use the standard notation $d(S):=\sum_{i\in S}d(i)$ and $w(S):=\sum_{i\in S}w(i)$.

\section{A \pE\ for \BagUFP}\label{sec:bag}

In this section we prove \Cref{thm:BagUFP-EPTAS}.
For the remaining of this section, fix a instance $I$ of \BagUFP\ and an error parameter $0<\eps <\frac{1}{2}$.  
Let the set of {\em heavy} tasks in $I$ be $$H = \left\{e \in E~|~ w(e) >  \frac{\eps \cdot \opt}{m}\right\}.$$ 
The remaining tasks $T \setminus H$ are {\em light}. 
Our first goal is to find the set of heavy tasks in a nearly-optimal solution. Notice that a naive enumeration takes ${n}^{{\Omega}(\frac{m}{\eps})}$ time, which is far from the running time of a \pE. To avoid this issue, we compute a (small enough) representative set, which is defined as follows. 

%
%
%
\begin{definition}
	\label{def:REP}
	For some $R \subseteq T$, we say that $R$ is an $\eps$-{\em representative set} of $I$ if there is a solution $S$
 of $I$ such that the following holds.
	\begin{enumerate}
		\item $S \cap H \subseteq R$. 
		\item $w\left(S\right) \geq (1-3\eps) \cdot \opt$.
	\end{enumerate} 
\end{definition}

 Define $q(\eps,m) = \ceil{ 4 m \cdot \eps^{-\ceil{\eps^{-1}}}}$ (the meaning of $q(\eps,m)$ becomes clear in \Cref{sec:repSet}). 

\begin{lemma}
	\label{lem:main}
There 
is an algorithm \textnormal{\textsf{RepSet}} that, given a \BagUFP\ instance instance $I=(G,u,T,P,d,w,{\cal B})$, $0<\eps<\frac{1}{2}$,  and $\tilde{\opt}\in [w(T)]$, in time $m^3 \cdot \eps^{-2} \cdot |I|^{O(1)}$ returns $R\subseteq T$ with $|R| \leq 3 \cdot m^3 \cdot \eps^{-2} \cdot q(\eps,m)$. Furthermore, if $\frac{\opt}{2} {<} \tilde{\opt} \leq \opt$, $R$ is an $\eps$-representative set of $I$.
\end{lemma} 

In \Cref{sec:nonProfitable} we use the representative set from \Cref{lem:main} to design a \pE~for \BagUFP. Then, in \Cref{sec:repSet} we prove \Cref{lem:main}. 

\subsection{A Representative Set Based $\pE$}
\label{sec:nonProfitable}

Given the representative set algorithm described  in \Cref{lem:main}, we obtain a \pE\ as follows (the pseudocode is given in \Cref{alg:EPTAS}).
\begin{algorithm}[h]
	\caption{$\textsf{p-EPTAS}(I,\eps)$}
	\label{alg:EPTAS}
	
	
	\SetKwInOut{Input}{input}
	
	\SetKwInOut{Output}{output}
	
	\Input{\BagUFP\ instance $I$ and an error parameter $0<\eps<\frac{1}{2}$.}
	
	\Output{ A $(1-7\eps)$-approximate solution $A$ for $I$.}
	
		$A \leftarrow \emptyset$.\label{step:init}
	
	\For{$\topt \in \left\{1,2,\ldots, 2^{{\floor{\log_2 \left(w(T)\right)}}}\right\}$}{
	
	
	Construct $R^{\tilde{\opt}} \leftarrow \textsf{RepSet} (I,\eps,\tilde{\opt})$.\label{step:representative}

	\For{$F \subseteq R^{\tilde{\opt}} \textnormal{ s.t. } |F| \leq m/\eps  \textnormal{ and } F \textnormal{ is a {feasible} solution for } I$\label{step:for}}{

		Compute a basic optimal solution $\lambda^{\tilde{\opt},F}$ for $\textnormal{LP}^{\topt}_F$.\label{step:basic}

				Define $L^{\tilde{\opt}}_F:=\left\{i \in T^{\topt}_F~\big|~ \lambda^{\topt,F}_{i} = 1\right\}$ and $A^{\tilde{\opt}}_F = L^{\tilde{\opt}}_F \cup F$.\label{step:Cf}

				\If{$w\left(A^{\tilde{\opt}}_F\right) > w(A)$\label{step:iff}}{
				
				$A \leftarrow A^{\tilde{\opt}}_F$.\label{step:improve}
				
				}

	}
	
}
	
	Return $A$.\label{step:retA}
\end{algorithm}
We consider the powers of two $\tilde{\opt}$ in the domain $[w(T)]$ (i.e., all values $\tilde{\opt} = 1, 2,4,\ldots, 2^{\floor{\log w(T)}}$). We apply the algorithm from \Cref{lem:main} with this parameter $\tilde{\opt}$ to obtain a set $R^{\tilde{opt}}$. Notice that, for $\frac{\opt}{2}<\tilde{\opt}\leq \opt$, $R^{\topt}$ is a representative set. Now we enumerate over all the feasible solutions $F\subseteq R^{\tilde{opt}}$ of cardinality at most $m/\eps$. For each such $F$, we compute a feasible solution $A^{\tilde{\opt}}_F$ (including $F$), and return the best such solution. 

It remains to describe how $A^{\tilde{\opt}}_F$ is computed. First of all, we define a reduced \BagUFP\ instance $I^{\tilde{\opt}}_F=(G,u_F,T^{\tilde{\opt}}_F,P,d,w,\cB_F)$ as follows. $\cB_F$ is the subset of input bags not containing any task in $F$. The set of tasks $T^{\tilde{\opt}}_F$ is given by the tasks of weight at most $\frac{2\eps}{m}\topt$ which are contained in the bags $\cB_F$. The capacity function $u_F$ is given by $u_F(e):=u(e)-\sum_{i\in F: e\in P(i)}d(i)$ (i.e., the residual capacity after accommodating the tasks in $F$). Observe that, for any feasible solution $L$ for $I^{\tilde{\opt}}_F$, $L\cup F$ is a feasible solution for the input problem. Indeed, the capacity constraints are satisfied and at most one task per bag can be selected.

Given the above instance $I^{\tilde{\opt}}_F$, we considering the following LP relaxation $LP^{\topt}_F$:
$$
\begin{aligned}
&\max && \sum_{i \in T^{\tilde{\opt}}_F} x_{i}   \cdot w(i) && (\,LP^{\topt}_F\,)\\
&\textnormal{s.t.} && \sum_{i \in T^{\tilde{\opt}}_F: e \in P(i)} x_{i}  \cdot d(i)  \leq  u_F(e) ~~~~&& \forall e\in E\\
&&& \sum_{i \in T^{\tilde{\opt}}_F \cap B}  x_{i} \leq 1 && \forall B \in \cB_F \\
&&& x_{i}\geq 0 && \forall i\in T^{\tilde{\opt}}_F
\end{aligned}
$$
We compute a basic optimal solution $\lambda^{\tilde{\opt},F}$ for the above LP. Let $L^{\topt}_F\subseteq T^{\topt}_F$ be the tasks such that $\lambda^{\tilde{\opt},F}_i=1$. We set $A^{\tilde{\opt}}_F=L^{\tilde{\opt}}_F\cup F$. This concludes the description of the algorithm.  

Obviously the above algorithm computes a feasible solution.
\begin{lemma}\label{lem:eptasBagUFP:feasibility}
\Cref{alg:EPTAS} returns a feasible solution.
\end{lemma}
\begin{proof}
Consider a given pair $(\tilde{\opt},F)$. Obviously $L^{\tilde{\opt}}_F$ is a feasible solution for the \BagUFP\ instance $I^{\tilde{\opt}}_F$. Indeed, the demand of the tasks in $L^{\tilde{\opt}}_F$ whose path contains a given edge $e$ is upper bounded by $\sum_{i \in T^{\topt}_F: e \in P(i)} \lambda^{\topt,F}_{i}  \cdot d(i)  \leq  u_F(e)$. Furthermore, for a given bag $B\in \cB_F$, at most one variable $\lambda^{\topt,F}_i$ with $i\in B$ can be equal to $1$, hence $|L^{\tilde{\opt}}_F\cap B|\leq 1$. Thus, as argued before, $A^{\tilde{\opt}}_F=L^{\topt}_F\cup F$ is a feasible solution for the input \BagUFP\ instance $I$. Since the returned solution $A$ is one of the feasible solutions $A^{\tilde{\opt}}_F$ (or the empty set, which is a feasible solution), $A$ is a feasible solution.
\end{proof}
It is also not hard to upper bound the running time.
\begin{lemma}\label{lem:eptasBagUFP:time}
\Cref{alg:EPTAS} runs in time $\left(\frac{3\cdot m^3}{\eps^2}\cdot q(\eps,m)\right)^{m/\eps}|I|^{O(1)}$.
\end{lemma}
\begin{proof}
Lines \ref{step:representative} and \ref{step:basic}-\ref{step:improve} can be performed in $|I|^{O(1)}$ time. Thus the overall running time is upper bounded by $|I|^{O(1)}$ multiplied by the number of possible pairs $(\topt,F)$. There are $O(\log w(T))=|I|^{O(1)}$ possible choices for $\topt$. For a fixed choice of $\topt$, one has $|R^{\topt}|\leq \frac{3m^3}{\eps^2}q(\eps,m)$. Since $F$ is a subset of $R^{\topt}$ of cardinality at most $m/\eps$, the number of possible choices for $F$ (for the considered $\topt$) is at most $2\left(\frac{3m^3}{\eps^2}q(\eps,m)\right)^{m/\eps}$. The claim follows.    
\end{proof}
It remains to bound the approximation factor of the algorithm. To this aim, we critically exploit the fact that each basic solution $\lambda^{\topt,F}$ is almost integral: more precisely, it has at most $2m$ non-integral entries. To prove that, we use a result in \cite{grandoni2010approximation} about the sparseness of matroid polytopes with $m$ additional linear constraints. 
\begin{lemma}\label{lem:eptasBagUFP:fractionality}
Each solution $\lambda^{\topt,F}$ computed by \Cref{alg:EPTAS} has at most $2m$ non-integral entries.
\end{lemma}
\begin{proof}
	The proof relies on matroid theory; for more details on the subject, we refer the reader to, e.g.,~\cite{schrijver2003combinatorial}. Consider $LP^{\topt}_F$ for any pair $(\topt,F)$ considered by the algorithm. Let $\tilde{LP}^{\topt}_F$ be the LP obtained from $LP^{\topt}_F$ by dropping the  $m$ capacity constraints $\sum_{i \in T^{\topt}_F: e \in P(i)} x_{i}  \cdot d(i)  \leq  u_F(e)$. $\tilde{LP}^{\topt}_F$ turns out to be the standard LP for a partition matroid (in particular, in an independent set at most one task per bag can be selected, where the bags induce a partition of the tasks). In \cite{grandoni2010approximation} it is shown that every basic solution (including an optimal one) for an LP obtained by adding $m$ linear constraints to the standard LP for any matroid (including partition ones) has at most $2m$ non-integral entries. Hence $\lambda^{\topt,F}$ satisfies this property. 
\end{proof}
\begin{lemma}\label{lem:eptasBagUFP:approximation}
The solution $A$ returned by \Cref{alg:EPTAS} satisfies $w(A)\geq (1-7\eps)\opt$.
\end{lemma}
\begin{proof}
It is sufficient to show that some solution $C^{\topt}_F$ has large enough profit. Consider the value of $\topt$ such that $\frac{\opt}{2}<\topt \leq \opt$. Notice that the algorithm considers exactly one such value since $1\leq \opt \leq w(T)$. We next show how to choose a convenient $F\subseteq R^{\topt}$.

 Observe that for the considered choice of $\topt$, $R^{\topt}$ is an $\eps$-representative set. Let $S$ be the solution for $I$ guaranteed by \Cref{lem:main} and \Cref{def:REP}. Recall that $w(S)\geq (1-3\eps)\opt$ and $S\cap H\subseteq R^{\topt}$. Since each $i\in S\cap H$ has $w(i)\geq \frac{\eps}{m}\opt$ by definition and since obviously $w(S\cap H)\leq w(S)\leq \opt$, it must be the case that $|S\cap H|\leq \frac{m}{\eps}$. This implies that there is an iteration of the algorithm (for the considered $\topt$) that has $F=S\cap H$: we will focus on that iteration. 

 We claim that $w(A^{\topt}_{S\cap H})=w(S\cap H)+w(L^{\topt}_F)\geq (1-7\eps)\opt$. Notice that each task $i\in S\setminus H$ has weight
$w(i)<\frac{\eps}{m}\opt<\frac{2\eps}{m}\topt$. Furthermore, by construction $i\in S\setminus H $ is contained in a bag in $\cB_{S\cap H}$. Hence $i\in T^{\topt}_{S\cap H}$, which implies $S\setminus H\subseteq T^{\topt}_{S\cap H}$. The feasibility of $S$ implies that $\sum_{i\in S\setminus H:e\in P(i)}d(i)\leq u_F(e)$ for every edge $e$, and $|(S\setminus H)\cap B|\leq 1$ for every $B\in \cB_{S\cap H}$. Therefore the integral solution $s$ which has $s_i=1$ for $i\in S\setminus H$ and $s_i=0$ for the remaining entries is a feasible solution for $LP^{\topt}_{S\cap H}$. Define $lp^{\topt}_{S\cap H}:=\sum_{i\in T^{\topt}_{S\cap H}}w(i)\cdot \lambda^{\topt,S\cap H}_i$ as the optimal LP value for $LP^{\topt}_{S\cap H}$. The feasibility of $s$ implies
$$
w(S\setminus H)=\sum_{i\in T^{\topt}_{S\cap H}}w(i)\cdot s_i\leq lp^{\topt}_{S\cap H}. 
$$
On the other hand,
$$
w(L^{\topt}_{S\cap H})\geq lp^{\topt}_{S\cap H} -2m\cdot \frac{2\eps}{m}\tilde{\opt}\geq lp^{\topt}_{S\cap H}-4\eps\cdot \opt.
$$
In the first inequality above we used the fact that $\lambda^{\topt,S\cap H}$ has at most $2m$ non-integral values (by \Cref{lem:eptasBagUFP:fractionality}), and that each $i\in T^{\topt}_{S\cap H}$ has $w(i)\leq \frac{2\eps}{m}\topt$ by construction. In the second inequality above we used the assumption that $\topt\leq \opt$. Putting everything together:
\begin{eqnarray*}
w(A^{\topt}_{S\cap H}) & = & w(L^{\topt}_{S\cap H})+w(S\cap H)\geq lp^{\topt}_{S\cap H}-4\eps\cdot \opt+ w(S\cap H)\\
& \geq & w(S\setminus H)-4\eps\cdot \opt+ w(S\cap H)=w(S)-4\eps\cdot \opt\geq (1-7\eps)\opt.
\end{eqnarray*}
\end{proof}
The proof of \Cref{thm:BagUFP-EPTAS} follows directly from \Cref{lem:eptasBagUFP:feasibility}, \Cref{lem:eptasBagUFP:time}, and \Cref{lem:eptasBagUFP:approximation}.

\subsection{Representative Set Construction}
\label{sec:repSet}

In this section, we construct a small $\eps$-representative set for the \BagUFP\ instance $I$; this gives the proof of \Cref{lem:main}. Let $\tilde{\opt} \in \left[w(T)\right]$ be a guess of the optimum value $\opt$. 
Recall that in \Cref{sec:nonProfitable} we are able to find $\tilde{\opt} \in \left(\frac{\opt}{2},\opt\right]$ using exponential search over the domain $\left[w(T)\right]$.  

We define a partition of the heavy tasks (and some tasks that are almost heavy) into {\em classes},
%
such that tasks of the same class have roughly the same weight and have the same subpath.  Specifically, let $\Phi =\{P(i)  \,|\,i\in T\}$ be the set of unique paths in the instance and define $\eta = \ceil{\log_{1-\eps} \left(\frac{\eps}{2 \cdot m}\right)}$ as a parameter describing the number of classes. For all 
$\varphi \in \Phi$ and $r \in \left[\eta\right]$
define the {\em class} of $\varphi$ and $r$ as  
\begin{equation}
	\label{eq:WWWW}
	\tilde{H} \left(\varphi,r\right) = \left\{i \in T~\bigg|~\frac{w(i)}{2 \cdot \tilde{\opt}} \in \bigg(\left(1-\eps\right)^r,\left(1-\eps\right)^{r-1} \bigg] \textnormal{ and } P(i) = \varphi\right\}. 
\end{equation} In simple words, a task $i$ belongs to class $\tilde{H} \left(\varphi,r\right) $ have weight roughly $ \left(1-\eps\right)^r \cdot 2 \cdot \tilde{\opt}$ and the subpath of $i$ is $\varphi$. 
Define $$\tilde{H} = \bigcup_{\varphi \in \Phi, r \in \left[\eta\right]}    \tilde{H}(\varphi,r)$$ as the union of classes. The parameter $\eta$ is carefully chosen so that the weight of every task $i \in \tilde{H}$ satisfies $w(i) \geq \frac{\eps \cdot \tilde{\opt}}{m}$, implying that $i$ is roughly heavy.  Since $\tilde{\opt} \in \left[\frac{\opt}{2},\opt\right]$, it follows that $H \subseteq \tilde{H}$ and that $\tilde{H}$ does not contain tasks with significantly smaller weight than $\frac{\eps \cdot \opt}{m}$ - that is the minimum weight allowed for heavy tasks.  
\begin{obs}
	\label{obs:Ht}
 $H \subseteq \tilde{H}$. 
\end{obs}

Let $\mathcal{D} = \{\tilde{H}(\phi,r)~|~\phi \in 
\Phi, r \in \left[\eta\right]\}$ be the set of classes. We use a simple upper bound on the number of classes. 
\begin{lemma}
	\label{lem:ProfitBound}
 $|\mathcal{D}| \leq 3 \cdot m^3 \cdot \eps^{-2}$. 
\end{lemma}

\begin{proof}
	Observe that 
	\begin{equation}
		\label{eq:ing}
		\log_{1-\eps} \left(\frac{\eps}{2 \cdot m}\right)  \leq 
		\frac{\ln \left(\frac{2 \cdot m}{\eps}\right)}{-\ln \left(1-\eps \right)} \leq \frac{ 2 \cdot m \cdot \eps^{-1}}{\eps} = 2 \cdot m \cdot \eps^{-2}.
	\end{equation} 
	The second inequality follows from  $x< -\ln (1-x), \forall x>-1, x \neq 0$, and $\ln (y) < y, \forall y>0$. Moreover, the number of subpaths $\varphi \in \Phi$ is bounded by $|\Phi| = {m \choose 2} \leq m^2$. Therefore,
	the number of classes 
	is bounded by
	\begin{equation}
		\label{eq:upper_bound_D}
		|\mathcal{D}| \leq m^2 \cdot \left(\log_{1-\eps} \left( \frac{\eps}{ 2 \cdot m}  \right)+1 \right) \leq 2 \cdot m \cdot \eps^{-2} \cdot m^2+m^2 = 2 \cdot m^3 \cdot \eps^{-2} +m^2 \leq 3 \cdot m^3 \cdot \eps^{-2}
	\end{equation}
	The first inequality follows from \eqref{eq:ing}. 
\end{proof}

\begin{algorithm}[h]
	\caption{$\textsf{RepSet}(I,\eps,\tilde{\opt})$}
	\label{alg:RepSet}
	
	
	\SetKwInOut{Input}{input}
	
	\SetKwInOut{Output}{output}
	
	\Input{\BagUFP instance $I$, an error parameter $0<\eps<\frac{1}{2}$, and $\tilde{\opt} \in \left[w(T)\right]$.}
	
	\Output{An $\eps$-representative set $R$ for $I$ (if $\tilde{\opt}\in \left[\frac{\opt}{2},\opt\right]$).} 



Initialize $R \leftarrow \emptyset$.

\ForAll{$\varphi \in \Phi$ \textnormal{and} $r \in \left[ \eta \right]$}{
	
	Let $\cB(\varphi,r) = \{B \in \cB~|~B \cap \tilde{H}(\varphi,r) \neq \emptyset\}$. 
	
	For every $B \in \cB(\varphi,r)$ define $i_B(\varphi,r) = \argmin_{i \in B \cap \tilde{H}(\varphi,r)} d(i)$. 
	
	Sort $\cB(\varphi,r)$ in non-decreasing order  $B_{1}(\varphi,r),\ldots, B_{\ell}(\varphi,r)$ 
	by~$d \left(i_B(\varphi,r)\right)~\forall B \in \cB(\varphi,r)$.\label{step:sorting}
	
	
	
	
	Define $a = \min \left\{  q(\eps,m), \left|\cB(\varphi,r)\right| \right\}$.
	
	Update $R \leftarrow R \cup \left\{ i_{B_1}(\varphi,r),\ldots, i_{B_a}(\varphi,r)\right\}$.\label{step:Ru}
	
}

Return $R$. 
\end{algorithm}

Our representative set construction 
is fairly simple. 
For each class $\tilde{H} \left(\varphi,r\right)$, consider the set of {\em active} bags $\cB(\varphi,r)$ for $\tilde{H}(\varphi,r)$ that contain at least one task in $\tilde{H} \left(\varphi,r\right)$. For every active bag $B \in \cB(\varphi,r)$ define the {\em representative} of $B$ in the class $\tilde{H}(\varphi,r)$ as the the task from the bag $B$ in the class $\tilde{H} \left(\varphi,r\right)$ of minimum demand (if there is more than one such task we choose one arbitrarily). We sort the active bags of the class in a non-decreasing order according to the demand of the representatives of the bags. Finally, we take the first $a$ representatives (at most one from each bag) according to this order, where $a$ is the minimum between the parameter $q(\eps,m)$ and the number of active bags for the class. The pseudocode of the algorithm is given in \Cref{alg:RepSet}. 



%
%

We give an outline of the proof of \Cref{lem:main}. Consider some optimal solution $\OPT$ for the instance. 
We partition the tasks in $\OPT$ into three sets: $L,J_{k^*}$, and $Q$ such that (i) the maximum weight of a task in $Q$ is at most $\eps$-times the minimum weight of a task in $L$; (ii) $L$ is small: $|L| \leq \frac{q(\eps,m)}{2}$; (iii)  
The weight of $J_{k^*}$ is small: $w(J_{k^*}) \leq \eps \cdot \opt$. 
To prove that $R$ is a representative set, we need to replace $H \cap \OPT$ with tasks from $R$. As a first step, we define a mapping $h$ from $\tilde{H} \cap \OPT$ to $R$, where each task $i \in \tilde{H} \cap \OPT$ is replaced by a task from the same class of a smaller or equal demand. For tasks $i \in \tilde{H} \cap \OPT$ such that $R$ contains a representative from the bag of $i$ in the class of $i$, we simply define $h(i)$ as this representative; for other tasks, we define the mapping via a bipartite matching on the remaining tasks and representatives. 

We define a solution $S$ satisfying the conditions of \Cref{def:REP} in two steps. First, we define initial solutions $S_1, S_2$. The solution $S_1$ contains the mapping $h(i)$ of every $i \in \tilde{H} \cap \OPT$ and the tasks in $L \setminus \tilde{H}$; the solution $S_2$ contains all tasks in $Q$ from bags that do not contain tasks from $S_1$. Finally, we define $S = S_1 \cup S_2$. By the properties of $L,J_{k^*}$, and $Q$ we are able to show that $S$ is roughly an optimal solution. Specifically, by (iii) discarding $J_{k^*}$ from the solution $S$ does not have a significant effect on the total weight of $S$. Additionally,by property (i) there is a large gap between the weights in $S_1$ and $S_2$; thus, combined with property (ii) we lose only a small factor due to tasks discarded from $Q$, and it follows that the weight of $S$ is $(1-O(\eps)) \cdot \opt$. 


\subsubsection*{Proof of \Cref{lem:main}}


We start with the running time analysis of the algorithm.

\begin{claim}
	\label{lem:RunningRepSet}
	The running time of \Cref{alg:RepSet} is bounded by $m^3 \cdot \eps^{-2} \cdot |I|^{O(1)}$  on input $I$, $\eps$, and $\tilde{\opt}$. Moreover, $|R| \leq 3 \cdot m^3 \cdot \eps^{-2} \cdot q(\eps,m)$. 
\end{claim}

\begin{claimproof}
	Each iteration of the {\bf for} loop of the algorithm can be trivially computed in time $|I|^{O(1)}$. In addition, the number of iterations of the {\bf for} loop is bounded by $3 \cdot m^3 \cdot \eps^{-2}$ using \Cref{lem:ProfitBound}. Therefore, the running time of the algorithm is bounded by $m^3 \cdot \eps^{-2} \cdot |I|^{O(1)}$. For the second property of the lemma, recall that the number of classes is bounded by $3 \cdot m^3 \cdot \eps^{-2}$ using \Cref{lem:ProfitBound}. By \Cref{step:Ru} of the algorithm, the number of tasks taken to $R$ from each class is at most $q(\eps,m)$. Therefore, $|R| \leq 3 \cdot m^3 \cdot \eps^{-2} \cdot q(\eps,m)$. 
\end{claimproof}

If $\tilde{\opt}\notin \left[\frac{\opt}{2},\opt\right]$, the proof immediately follows from \Cref{lem:RunningRepSet}. Thus, for the following assume that $\tilde{\opt}\in \left[\frac{\opt}{2},\opt\right]$. Let $\OPT \subseteq T$ be an optimal solution for $I$. 
Let $w^* = \frac{\eps \cdot \tilde{\opt}}{m}$ be a lower bound on the minimum weight of a task in $\tilde{H}$. We partition a subset of the tasks in $\OPT \setminus \tilde{H}$ with the highest weights 
into $N = \ceil{\eps^{-1}}$ disjoint sets.
For all $k \in [N]$ define the $k$-th set as 
\begin{equation}
	\label{J_k}
	J_k = \left\{i \in \OPT \setminus \tilde{H}~\big|~ w(i) \in \big( \eps^{k} \cdot w^* ,  \eps^{k-1} \cdot w^* \big] \right\}.		
\end{equation} 
Let $k^* = \argmin_{k\in [N]} w(J_k)$. By \eqref{J_k} the sets $J_1,\ldots, J_N$ are 
$N \geq \eps^{-1}$ disjoint sets (some of them may be empty); thus, $w(J_{k^*}) \leq \eps \cdot \opt$. 
Define $$L~= \left(\OPT \cap \tilde{H}\right) \cup \bigcup_{k \in  [k^*-1]} J_k$$ as the subset of all tasks in $\OPT$ 
of weight greater than $\eps^{k^*-1} \cdot w^*$, and define $Q = \OPT \setminus (L \cup J_{k^*})$ as the remaining tasks in $\OPT$ excluding $J_{k^*}$.  
We use the following auxiliary claim. 
\begin{claim}
	\label{claim:Iq}
	$|L| \leq  \frac{q(\eps,m)}{2}$. 
\end{claim}
\begin{claimproof}
	If $L = \emptyset$ the claim trivially follows. Otherwise, 
	\begin{equation}
		\label{contradiction1}
		\begin{aligned}
			|L| \leq \sum_{i \in L} \frac{w(i)}{\eps^{k^*-1} \cdot w^*} 
			= \frac{w(L)}{\eps^{k^*-1} \cdot w^*} 
			\leq  \frac{\opt}{\eps^{k^*-1} \cdot w^*} 
		\end{aligned}
	\end{equation} 
 The first inequality holds since $w(i) \geq \eps^{k^*-1} \cdot w^*$
	for all $i \in L$.  The second inequality follows from the fact that  $L \subseteq \OPT$; thus, $L$ is a solution for $I$. 
 Thus, by~\eqref{contradiction1} and the definition of $w^*$
	\begin{equation*}
		\begin{aligned}
			|L| \leq  \frac{\opt}{\eps^{k^*-1} \cdot 2 \tilde{\opt} \cdot  \frac{\eps}{2 \cdot m} } \leq \frac{\opt}{\eps^{k^*-1} \cdot \opt \cdot  \frac{\eps}{2 \cdot m} } 
			= \frac{2 \cdot m}{\eps^{k^*}} 
			\leq \frac{2 \cdot m}{\eps^{N}} 
			\leq  \frac{q(\eps,m)}{2}.
		\end{aligned}
	\end{equation*} The second inequality holds since we assume that $\tilde{\opt} \geq \frac{\opt}{2}$. 
\end{claimproof}

Let $R$ be the set returned by the algorithm. In the following, we show the existence of a solution $S$ such that $S \cap H \subseteq R$ and $w(S) \geq (1-3 \cdot \eps) \cdot \opt$; this gives the statement of the lemma by \Cref{def:REP}. To construct $S$, we first define a mapping $h$ from $\tilde{H} \cap \OPT$ to $R$. For a subpath $\varphi \in \Phi$ and $r \in \left[ \eta \right]$, recall the set of active bags $\cB(\varphi,r)$ and the representatives $i_B(\varphi,r)$ for all $B \in \cB(\varphi,r)$ (see \Cref{alg:RepSet}).  


\omitmac{
	For a subpath $\varphi \in \Phi$, $r \in \left[ \eta \right]$, and $i \in \OPT \cap \tilde{H}(\varphi,r)$ such that $R \cap \tilde{H}(\varphi,r) \cap B^i \neq \emptyset$ define 
	\begin{equation}
		\label{eq:h(i)}
		h(i) = \argmin_{t \in R \cap \tilde{H}(\varphi,r) \cap B^i} d(t)
	\end{equation} as the {\em mapping} of $i$ to $R$. The next claim summarizes the crucial properties of the above mapping. 
	
	\begin{claim}
		\label{claim:h(i)}
		Let $\varphi \in \Phi$, $r \in \left[ \eta \right]$, and $i \in \OPT \cap \tilde{H}(\varphi,r)$ such that $R \cap \tilde{H}(\varphi,r) \cap B^i \neq \emptyset$. Then, $h(i)$ there is exactly one $t \in R \cap \tilde{H}(\varphi,r) \cap B^i$ such that $h(i) = t$ and it holds that $d(h(i)) \leq d(i)$. 
	\end{claim}
	\begin{claimproof}
		Since $R \cap \tilde{H}(\varphi,r) \cap B^i \neq \emptyset$, by \Cref{step:Ru} of the algorithm there is $j \in \left\{ 1,\ldots,\min(q(\eps,m),|\cB(\varphi,r)|)\right\}$ such that $B^i = B_{j}(\varphi,r)$. Therefore, by \Cref{step:Ru} we have $h(i) = \argmin_{t \in B_{j}(\varphi,r) \cap \tilde{H}(\varphi,r)} d(t)$. Therefore, $h(i)$ is well defined. Thus, since $i \in B_{j}(\varphi,r) \cap \tilde{H}(\varphi,r)$ it follows that $d(h(i)) \leq d(i)$ by the above. 
	\end{claimproof}
}

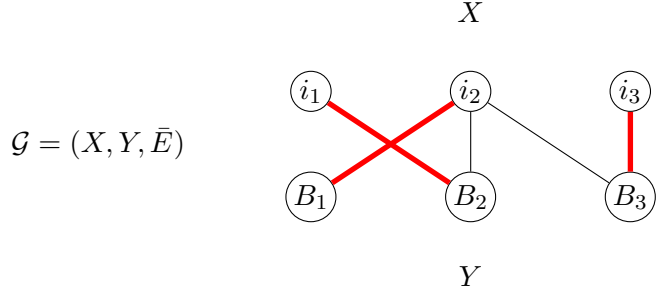
\begin{figure}
		\centering
		\begin{tikzpicture}[scale=1.4, every node/.style={draw, circle, inner sep=1pt}]
			\node (l1a) at (5.5,-0.5) {$\textcolor{black}{i_2}$};
			\node (1) at (4,-0.5) {$\textcolor{black}{i_1}$};
			\node (l2a) at (7,-0.5) {$\textcolor{black}{i_3}$};
%
	\node (2) at (4,-1.5) {$B_1$};

\node (r1a) at (5.5,-1.5) {$B_2$};
\node (r2a) at (7,-1.5) {$B_3$};
			\draw (l1a) -- (r1a);
			\draw[line width=2pt, color=red] (l2a) -- (r2a);

			\draw (r2a) -- (l1a);
			\draw[line width=2pt, color=red] (2) -- (l1a);
			\draw[line width=2pt, color=red] (1) -- (r1a);

%
			
			\node[draw=none] at (2, -1) {$\mathcal{G} = (X,Y,\bar{E})$};
			
				\node[draw=none] at (5.5, 0.25) {$X$};
				
								\node[draw=none] at (5.5, -2.25) {$Y$};
			
%
		\end{tikzpicture}
		\caption{\label{fig:X} An illustration of the graph $\mathcal{G}$ and the maximum matching $M$ (in red). Every edge $(i,B)$ in the graph indicates that bag $B$ belongs to $\textsf{fit}(i)$; that is, the representative from $B$ in the class of $i$ belongs to $R$ and the demand of this representative is at most the demand of $i$. Note that even though $i_1$ and $i_2$ are both connected to bag $B_2$, $i_1$ and $i_2$ may belong to different classes.}
	\end{figure}

For the simplicity of the notation, for $\varphi \in \Phi$, $r \in \left[ \eta \right]$, and $i \in \tilde{H}(\varphi,r)$ let $\tilde{H}^i = \tilde{H}(\varphi,r)$ be the class to which $i$ belongs and let $r_i = r$; moreover, for $B \in \cB$ such that $i \in B$ define $B^i = B$ as the bag containing $i$. 
We first consider tasks $i$ in $\OPT \cap \tilde{H}$ whose bag does not have a representative in $R$ from the class of $i$, i.e., $R \cap \tilde{H}^i \cap B^i = \emptyset$. Define this set of tasks as \begin{equation}
	\label{eq:Xh}
	X = \left\{        i \in \OPT \cap \tilde{H}    ~\bigg|~R \cap \tilde{H}(\varphi,r) \cap B^i = \emptyset \right\}.
\end{equation} The above set $X$ contains all tasks $i \in \OPT \cap \tilde{H}$ whose corresponding bag does not have a representative in $R$ from the class of $i$. 
We define a bipartite graph, in which $X$ is one side of the graph. The other side of the graph is  
\begin{equation}
	\label{eq:Yh}
	Y = \cB \setminus \left\{  B \in \cB ~\big|~ \exists i \in L \text{ s.t. } B = B^i \right\}. 
\end{equation}
In words, $Y$ describes all {\em available} bags, the collection of all bags that do not contain a task in $L$. Define the bipartite graph $\mathcal{G} = (X,Y,\bar{E})$ such that the set of edges is defined as follows. For some 
$i \in X$ let 
\begin{equation}
	\label{eq:FitP,R}
	\textsf{fit}(i) = \left\{  B \in Y~\bigg|~ B \cap R \cap  \tilde{H}^i \neq \emptyset   \textnormal{ and }   d \left(i_{B} (P(i),r_i)\right) \leq d(i) \right\}. 
\end{equation}
The set $\textsf{fit}(i)$ describes all bags that can potentially matched to $i$; these bags have a representative from the class $\tilde{H}^i = \tilde{H}(P(i),r_i)$ that contain $i$ and the representative of the bag have a smaller or equal demand w.r.t. $i$. Now, a task $i$ can be matched to a bag $B$ only if $B \in \textsf{fit}(i)$, i.e., define  
\begin{equation}
	\label{eq:BarE}
	\begin{aligned}
		\bar{E} = \left\{        \left(i,B \right) \in X \times Y ~\bigg|~
		B \in 	\textsf{fit}(i) \right\}.
	\end{aligned}
\end{equation} Let $M$ be a maximum matching in $\mathcal{G}$. We give an illustration of the above construction in \Cref{fig:X}. We show that $M$ matches all vertices in $X$. 

\begin{claim}
	\label{claim:MatchingM}
	For every $i \in X$ there is $B \in Y$ such that $(i,B) \in M$.
\end{claim}
\begin{claimproof}
	Assume towards a contradiction that there is $i \in X$ such that for all $B \in Y$ it holds that $(i,B) \notin M$. Let $\varphi \in \Phi$ and $r \in \left[ \eta \right] $ such that $\tilde{H}^i = \tilde{H}(\varphi,r)$. Since $i \in X$, by \eqref{eq:Xh} it holds that $R \cap \tilde{H}(\varphi,r) \cap B^i = \emptyset$. Intuitively, this means that the algorithm preferred other bags over $B^i$ in the selection of representatives for class   $\tilde{H}(\varphi,r)$. Therefore, by \Cref{step:sorting,step:Ru} of the algorithm, there are $q(\eps,m)$ distinct bags $B_1 = B_1 (\varphi,r),\ldots, B_{q(\eps,m)} = B_{q(\eps,m)}(\varphi,r)$ such that for all $j \in \left[q(\eps,m)\right]$ it holds that $i_{B_j}(\varphi,r) \in R$ and $d \left(i_{B_j}(\varphi,r)\right) \leq d(i)$. 
	Thus, for all $j \in \left[q(\eps,m)\right]$ it holds that $\left(i,B_j\right) \in \bar{E}$ by \eqref{eq:FitP,R} and \eqref{eq:BarE}. In addition, 
	\begin{equation}
		\label{eq:X<q}
		|M| \leq |X| \leq |L| \leq \frac{q(\eps,m)}{2} < q(\eps,m). 
	\end{equation} The first inequality holds since $M$ is a matching in $\mathcal{G}$ and $X$ is one side of a bipartition of $\mathcal{G}$. The second inequality holds since $X \subseteq L$ by \eqref{eq:Xh} and the definition of $L$. The third inequality follows from \Cref{claim:Iq}. The last inequality holds since $q(\eps,m) \geq 2$ assuming $0<\eps<\frac{1}{2}$ and $m \geq 1$. By \eqref{eq:X<q} there is $j \in \left[q(\eps,m)\right]$ such that for all $t \in X$ it holds that $\left(t,B_j \right) \notin M$. In particular, $\left(i,B_j\right) \notin M$ and recall that $\left(i,B_j\right) \in \bar{E}$. Therefore, $M \cup \left(i,B_j\right)$ is a matching in $\mathcal{G}$ in contradiction that $M$ is a maximum matching in $\mathcal{G}$. 
\end{claimproof}

For every $i \in X$ define $M_i = B$ such that $(i,B) \in M$, i.e., $M_i$ is the bag matched to $i$ in $M$. By \Cref{claim:MatchingM} it holds that each task in $X$ is matched and every bag is matched at most once.   
We define the mapping $h$ from $\tilde{H} \cap \OPT$ to $R$. Define $h: \tilde{H} \cap \OPT \rightarrow R$ such that for all $i \in \tilde{H} \cap \OPT$:
\begin{equation}
	\label{eq:h(i)}
	h(i) = \begin{cases}
		i_{B^i} \left(P(i),r_i\right), & \text{if }  B^i \cap R \cap  \tilde{H}^i \neq \emptyset\\
			i_{M_i} \left(P(i),r_i\right) , & \text{else } 
	\end{cases}
\end{equation} In words, a task $i \in \tilde{H} \cap \OPT$ is mapped to a task $h(i)$ such that if the bag of $i$ contains a representative in $R$ in the class of $i$ - then $h(i)$ is this representative; otherwise, $h(i)$ is the representative of the bag $M_i$ matched to $i$ by the matching $M$. Clearly, $h$ is well defined by \Cref{claim:MatchingM}. We list immediate properties of $h$.
\begin{obs}
	\label{obs:hProp}
	The function $h$ satisfies the following.
	\begin{itemize}
		\item For every $i \in \tilde{H} \cap \OPT$ it holds that $d(h(i)) \leq d(i)$ and $\tilde{H}^{h(i)} = \tilde{H}^{i}$. 
		\item For every $i,j \in \tilde{H} \cap \OPT$, $i \neq j$, it holds that $B^{h(i)} \neq B^{h(j)}$. 
		\item For every $i \in \tilde{H} \cap \OPT$ and $t \in L \setminus \tilde{H}$ it holds that $B^{h(i)} \neq B^{t}$. 
	\end{itemize}
\end{obs} The first property follows from the definition of the graph $\mathcal{G}$ and the definition of the bag representatives in \Cref{alg:RepSet}. The second and third properties hold since $\OPT$ takes at most one task from each bag and using the definition of $\mathcal{G}$. 
We can finally define the solution $S$ that satisfies the conditions of \Cref{def:REP}. Define 
\begin{equation}
	\label{eq:U1}
	S_1 = \left\{  h(i)~|~i \in \tilde{H} \cap \OPT \right\} \cup \left(L \setminus \tilde{H}\right) 
\end{equation}
and
\begin{equation}
	\label{eq:U2}
	S_2 =  \left\{   i \in Q~|~B^i \neq B^t~\forall t \in S_1\right\}. 
\end{equation}
Define $S = S_1 \cup S_2$. 
We show that $S$ satisfies the conditions of  \Cref{def:REP}. As an immediate property of the construction we have the following.
\begin{obs}
	\label{obs:one-to-one}
	$h$ is a one-to-one function from $\tilde{H} \cap \OPT$ to $S \cap \tilde{H}$. 
\end{obs}

We use the above to prove the feasibility of $S$. 

\begin{claim}
	\label{claim:IsSolution}
	$S$ is a solution for $I$. 
\end{claim}

\begin{claimproof}
	We show that $S$ satisfies the bag constraints. Let $B \in \cB$. Since $\OPT$ is a solution for $I$, there is at most one $i \in B \cap \OPT$. We consider four cases depending on the task $i$.
	
	\begin{enumerate}
		\item 	If $i \in \tilde{H}$ and 
		$R \cap \tilde{H}^i \cap B^i \neq \emptyset$. Then, $h(i) \in B$ by \eqref{eq:h(i)} and for all $t \in S_1 \setminus \{h(i)\}$ it holds that $t \notin B$ by \Cref{obs:hProp}. Furthermore, for all $t \in S_2$ it holds that $t \notin B$ by  \eqref{eq:U2}. Thus, $|B \cap S| \leq 1$.
		
		\item 	If 
		$i \in \tilde{H}$ and $R \cap \tilde{H}^i \cap B^i = \emptyset$. Then, as $\tilde{H} \subseteq L$ it holds that $i \in L$; thus, $B \notin Y$ by \eqref{eq:Yh}. Therefore, by \eqref{eq:U1} we conclude that $|B \cap S_1| = 0$; thus, 
		$$|B \cap S| = |B \cap S_2| \leq |B \cap Q| \leq |B \cap \OPT| \leq 1.$$ The equality holds since $|B \cap S_1| = 0$. The first inequality follows from \eqref{eq:U2}. The last inequality holds since $\OPT$ is a solution.  
		
		\item If $i \in L \setminus \tilde{H}$. Then, by \eqref{eq:U1} and \eqref{eq:U2} it holds that $|B \cap S| = |B \cap {i}| = 1$. 
		
		\item  If $i \in Q$. Then, there are two sub cases. If $i \in S_2$, by \eqref{eq:U2} for all $t \in S_1$ it holds that $B \neq B^t$; thus, as $|Q \cap B| \leq |\OPT \cap B| \leq 1$ it follows that $|S \cap B| \leq 1$. Otherwise, $i \notin S_2$; then, by \eqref{eq:U2} it holds that $$|B \cap S| = |B \cap S_1| \leq 1.$$ 
		The inequality follows from \Cref{obs:hProp}. 
	\end{enumerate} By the above we conclude that $S$ satisfies all bag constraints. It remains to prove that $S$ satisfies the capacity constraints of all edges. For $e \in E$
	
	\begin{equation*}
		\label{eq:sB2}
		\begin{aligned}
			\sum_{i \in S \text{ s.t. } e \in P(i)} d(i) ={} &  \sum_{i \in S \cap \tilde{H} \text{ s.t. } e \in P(i)} d(i)+ 	\sum_{i \in S \setminus \tilde{H} \text{ s.t. } e \in P(i)} d(i)\\
			={} &  \sum_{i \in \OPT \cap \tilde{H} \text{ s.t. } e \in P(i)} d(h(i))  + 	\sum_{i \in S \setminus \tilde{H} \text{ s.t. } e \in P(i)} d(i)\\
			\leq{} &  \sum_{i \in \OPT \cap \tilde{H} \text{ s.t. } e \in P(i)} d(h(i))  + 	\sum_{i \in \OPT \setminus \tilde{H} \text{ s.t. } e \in P(i)} d(i)\\
			\leq{} &  \sum_{i \in \OPT \cap \tilde{H} \text{ s.t. } e \in P(i)} d(i)  + 	\sum_{i \in \OPT \setminus \tilde{H} \text{ s.t. } e \in P(i)} d(i)\\
			={} & \sum_{i \in \OPT \text{ s.t. } e \in P(i)} d(i)\\
			\leq{} & u(e).
		\end{aligned}
	\end{equation*} The second equality holds since $h$ is a one-to-one mapping from $\OPT \cap \tilde{H}$ to $S \cap \tilde{H}$ by \Cref{obs:one-to-one}. The first inequality holds since $S \setminus \tilde{H} \subseteq \OPT \setminus \tilde{H}$ by \eqref{eq:U1} and \eqref{eq:U2}. The second inequality holds since $d(h(i)) \leq d(i)$ for all $i \in \OPT \cap \tilde{H}$ by \Cref{obs:hProp}. The last inequality holds since $\OPT$ is a solution for $I$. 
\end{claimproof} 

Observe that there is a substantial gap in weight between tasks in $L$ and tasks in $Q$. We use this gap in the following auxiliary claim. 
%
\begin{claim}
	\label{clam:weightBound1}
	$w\left( Q \setminus S\right) \leq \eps \cdot \opt$.
\end{claim}
\begin{claimproof} 
	Observe that 
	\begin{equation}
		\label{eq:menS}
		\begin{aligned}
			| Q \setminus S| =  	| Q \setminus S_2| =  \left| \left\{   i \in Q~|~\exists t \in S_1 \text{ s.t. } B^i = B^t\right\} \right|
			\leq |S_1|
			= |L|. 
		\end{aligned}
	\end{equation}
	The inequality holds since $Q$
 satisfies the bag constraints (i.e., $|Q \cap B| \leq 1$ for all $B \in \cB$); thus, for each $t \in S_1$ there can be at most one $i \in Q$ such that $B^i = B^t$ (and only in this case $i$ is discarded from $S_2$). The last equality holds since $h$ is a one-to-one mapping from $\OPT \cap \tilde{H}$ to $S \cap \tilde{H}$ by \Cref{obs:one-to-one} and since $L \setminus \tilde{H}$ belongs both to $S_1$ and $L$. Hence, 
	\begin{equation*}
		\begin{aligned}
			w( Q \setminus S) \leq{} & | Q \setminus S| \cdot \eps^{k^*} \cdot w^* \leq  |L| \cdot \eps^{k^*} \cdot w^* \leq  \eps \cdot w(L) \leq  \eps \cdot w(\OPT) =  \eps \cdot \opt.
		\end{aligned}
	\end{equation*} 
	
	The first inequality holds since $w(i) \leq \eps^{k^*} \cdot w^*$ for all $i \in Q$. 
	The second inequality follows from~\eqref{eq:menS}. The third inequality holds since $w(i) > \eps^{k^*-1} \cdot w^*$ for all $i \in L$. the last inequality holds since $L \subseteq \OPT$. 
\end{claimproof} 	


The following claim shows that $S$ satisfies the total weight required by \Cref{def:REP}.
\begin{claim}
	\label{clam:weightBound2}
	$w\left(S\right) \geq (1-3\eps) \cdot \opt$.
\end{claim} 

\begin{claimproof} 
	We first give a lower bound to the weight of $S_1$. 
	\begin{equation}
		\label{eq:proofProfit1}
		\begin{aligned}
			w(S_1) ={} & w\left((L \setminus \tilde{H}) \cup  \left\{  h(i)~|~i \in \tilde{H} \cap \OPT \right\} \right) \\
			={} & w\left(L \setminus \tilde{H}\right)+ \sum_{i \in \tilde{H} \cap \OPT} w(h(i)) \\
			\geq{} & w\left(L \setminus \tilde{H}\right)+ \sum_{i \in \tilde{H} \cap \OPT} (1-\eps) \cdot w(i) \\
			\geq{} & (1-\eps) \cdot w(L).
		\end{aligned}
	\end{equation} 
	The inequality holds since  for all $i \in \OPT \cap \tilde{H}$ it holds that $\tilde{H}^i = W^{h(i)}$ by \Cref{obs:hProp}; thus, by \eqref{eq:WWWW} it follows that $w(h(i)) \geq (1-\eps) \cdot w(i)$. For the last inequality, recall that $\tilde{H} \subseteq L$. Moreover, 
	\begin{equation}
		\label{eq:proofProfit2}
		\begin{aligned}
			w(S_2) {} & = w(Q) - w(Q \setminus S) \geq w(Q) - \eps \cdot \opt \geq (1-\eps) \cdot w(Q) - \eps \cdot \opt. 
		\end{aligned}
	\end{equation} 
	
	The first equality holds since $S_2 \subseteq Q$. The first inequality follows from \Cref{clam:weightBound1}. By \eqref{eq:proofProfit1} and \eqref{eq:proofProfit2} we have
	
	\[	
	\begin{array}{ll}
		w(S) & = w(S_1)+w(S_2) \\
		& \geq (1-\eps) \cdot w(L \cup Q)-\eps \cdot \opt 
		\\ & = (1-\eps) \cdot w(\OPT \setminus J_{k^*})-\eps \cdot \opt \\
		&  \geq (1-\eps) \cdot (1-\eps) \cdot \opt-\eps \cdot \opt\\
		&  \geq (1-3\eps) \cdot \opt. 
	\end{array}
	\]
	
	The second inequality holds since $w(J_{k^*}) \leq \eps \cdot \opt$. 
\end{claimproof} 

Observe that $H \subseteq \tilde{H}$ by \Cref{obs:Ht}. 
Moreover, $S \cap \tilde{H} = S_1 \cap \tilde{H} \subseteq R$ by \eqref{eq:U1} and \eqref{eq:U2}. Thus, $S \cap H \subseteq R$.  
By \Cref{claim:IsSolution} and \Cref{clam:weightBound2}, it follows that $R$ is a representative set. 

The proof follows from  \Cref{claim:IsSolution}, \Cref{clam:weightBound2}, and \Cref{lem:RunningRepSet}. \qed

\newpage

\section{A Faster \pE\ for \UFP}
\label{sec:UFPalg}

In this section we prove \Cref{thm:eptasUFP}.
Let $(G,u,T,P,d,w)$ be a \UFP\ instance. For simplicity, we next assume that $1/\eps$ is integer and that $n\gg 1/\eps$. Recall that $\Phi =\{P(i)  \,|\,i\in T\}$ is the set of unique paths in the instance, and for every $\varphi \in \Phi$ we use 
$T_{\varphi}= \{i\in T\,|\,P(i)=\varphi\}$ to denote the set of tasks with path $\varphi$. Observe that $|\Phi|\leq \frac{1}{2}\cdot m\cdot(m+1)$.

See \Cref{alg:UFP2} for a pseudocode description of our approach. 
\begin{algorithm}[h]
	\caption{\pE\ for \UFP.}
	\label{alg:UFP2}
	\SetKwInOut{Input}{input}
	\SetKwInOut{Output}{output}
 	\SetKwInOut{Notations}{Notations}

	\Input{\UFP\ instance $\II= (G,u,T,P,d,w)$ and a parameter $0<\eps<0.1$} 
	
	\Output{A feasible solution \APX\ for the instance $\II$}
 \Notations{ Here $w_{\max}=\max_{i\in T}w(i)$ and $p_\varphi:=\sum_{i\in T_\varphi}p(i)$ for all $\varphi\in \Phi$.}
	
	Define $p(i) \,\leftarrow\, \floor{\frac{n \cdot  w(i)}{\eps \cdot w_{\max}}}$ for every $i\in T$.\label{step:roundProfit}
	
    For all $\varphi \in \Phi$ compute $\Sol_\varphi$ for the Knapsack instance  $(T_{\varphi}, d, p,\min_{e\in \varphi} u(e))$.
\label{step:dynamic} 
	
	$\APX \leftarrow \,\emptyset$.

	\For{all the powers $\tilde{\opt}$ of $(1+\eps)$ in $\left[1,\frac{n^2}{\eps}\right]$}{	
	\For{all non-negative integers $(X_{\varphi})_{\varphi\in \Phi}$ such that $\sum_{\varphi \in \Phi} X_\varphi \leq \abs{\Phi}\cdot \frac{1+\eps}{\eps}$\label{step:loop}}{

            Set $\tilde{\opt}_{\varphi} =X_{\varphi }\cdot \frac{\eps}{|\Phi|}\cdot \tilde{\opt} $ 
	 for all $\varphi\in \Phi$.\label{step:setTildeOpt}\\
		$\APX'\,\leftarrow \, \bigcup_{\varphi \in \Phi} \Sol_{\varphi}\left( \min\left \{p_{\varphi}, \ceil{\tilde{\opt}_\varphi}\right\}\right)$.\label{step:findR2}
		
		{\bf if} $\APX'$ is a feasible solution for $\II$ and ${p}(\APX')\geq {p}(\APX)$ {\bf then} $\APX\leftarrow  \APX'$.\label{step:update2}		
	}
 }
	
	Return $\APX$.
\end{algorithm}

We start to perform a standard rounding of the weights (similar to several other packing problems) so that they are positive integers in a polynomially bounded range. 
Let $w_{\max}=\max_{i\in T} w_i$ be the maximum weight of any task. Observe that, since w.l.o.g. each task alone induces a feasible solution, one has $\opt\geq w_{\max}$. 
We replace each weight $w(i)$ with $p(i):=\floor{\frac{n\cdot w(i)}{\eps \cdot w_{\max}}}$.
A standard calculation  shows that an optimum solution $\OPT'$ computed w.r.t. the modified weights $p$ is a $(1-\eps)$-approximate solution w.r.t. the original problem. Now the (rounded) weights are in the range $\left[\frac{n}{\eps}\right]$. With the obvious notation, for $S\subseteq T$, we will denote $p(S):=\sum_{i\in S}p(i)$.


Now we proceed by describing the two main phases of our \pE. In the first phase we consider each path $\varphi\in \Phi$, and define a Knapsack instance $K_{\varphi}=(T_{\varphi},d,p,\min_{e\in \varphi} u(e))$. Here $T_{\varphi}$ is the set of items that can be placed in the knapsack, $d(i)$ and $p(i)$ are the size and profit of item $i\in T_{\varphi}$, resp., and $\min_{e\in \varphi} u(e)$ is the size of the knapsack. We solve this instance $K_{\varphi}$ using the standard algorithm for Knapsack based on dynamic programming. In more detail, this algorithm defines a dynamic programming table $\Sol_{\varphi}$ indexed by the possible values $p'\in [p_{\varphi}]$ of the profit, where $p_{\varphi}:=\sum_{i\in T_\varphi}p(i)$. At the end of the algorithm, for each such $p'$, $\Sol_{\varphi}(p')$ contains a subset of items (in $T_{\varphi}$) of minimum total size (or, equivalently, demand) whose profit is at least $p'$\footnote{In a more standard version of the algorithm $\Sol_{\varphi}(p')$ would contain a minimum size solution of profit \emph{exactly} $p'$, or a special character if such solution does not exist. However, it is easy to adapt the algorithm to rather obtain the desired values.}. Notice that $T_\varphi$ satisfies $p(T_\varphi)=p_\varphi\geq p'$, hence all the table entries $\Sol_\varphi(p')$ are well defined. We also remark that certain table entries may contain a solution of total demand larger than $\min_{e\in \varphi}u(e)$, hence such entries will never be used to construct a feasible \UFP\ solution. 
Computing the dynamic tables for all $\varphi\in \Phi$ takes time 
$$
\sum_{\varphi\in \Phi}O(|T_{\varphi}|\cdot p_{\varphi})\,\leq\, \sum_{\varphi\in \Phi}O(|T_{\varphi}|)\cdot \sum_{\varphi\in \Phi}O(p_{\varphi})\,=\,O(|T|)\cdot O(p(T))\,\leq\, O\left(\frac{n^3}{\eps}\right).
$$ 
We store these dynamic tables for later use. 

At this point the second phase of the algorithm starts. Let $\opt'=p(\OPT')$, where $\OPT'$ is an optimal solution for the $\UFP$ instance with the rounded weights $p$ (i.e.,  $(G,u,T,P,d,p)$). Observe that $\opt'\in \left[\frac{n^2}{\eps}\right]$. Let $\tilde{\opt}$ be a power of $(1+\eps)$ such that $\frac{\opt'}{1+\eps}< \tilde{\opt}\leq \opt'$. We can find this value by trying all the $O\left(\log_{1+\eps}\frac{n^2}{\eps}\right)=O\left(\frac{1}{\eps}\cdot \log n\right)$ possibilities. Define $\OPT'_{\varphi}:=\OPT'\cap T_{\varphi}$ and $\opt'_\varphi=p(\OPT'_{\varphi})$. For each $\varphi\in \Phi$ we guess the largest multiple $\tilde{\opt}_\varphi=X_{\varphi}\cdot \frac{\eps}{|\Phi|}\cdot \tilde{\opt}$ of $\frac{\eps}{|\Phi|}\cdot \tilde{\opt}$ which is upper bounded by $\opt'_{\varphi}$. Again by guessing we mean trying all the possible combinations. Obviously a valid guess must satisfy 
$$
\frac{\eps}{|\Phi|}\cdot \tilde{\opt}\cdot\sum_{\varphi\in \Phi}X_{\varphi} =
\sum_{\varphi\in \Phi}\tilde{\opt}_\varphi\leq \sum_{\varphi\in \Phi}\opt'_\varphi = \opt'\leq (1+\eps)\tilde{\opt},
$$ 
hence $\sum_{\varphi\in \Phi}X_{\varphi}\leq Y:=\lfloor \frac{1+\eps}{\eps}|\Phi|\rfloor$. Thus it is sufficient to generate all the ordered sequences of $|\Phi|$ non-negative integers whose sum is at most $Y$. As we will argue, the number of such sequences is sufficiently small.

Given a guess $\{\tilde{\opt}_\varphi\}_{\varphi\in \Phi}$, we compute a tentative solution $\APX':=\cup_{\varphi\in \Phi}\Sol_{\varphi}(\min\{p_\varphi,\lceil \tilde{\opt}_\varphi\rceil\})$ using the pre-computed dynamic tables. {Notice that, for a valid guess of $\tilde{\opt}_\varphi$, by integrality we also have $\lceil \tilde{\opt}_\varphi\rceil\leq \opt'_\varphi$. Upper bounding with $p_\varphi\geq \opt'_\varphi$ guarantees that the algorithm only uses well-defined table entries.} Among the solutions $\APX'$ which are feasible, we return one $\APX$ of maximum profit $p(\APX)$. This concludes the description of the algorithm. 

We can further improve the running time as follows. Let us compute and store the values $d(\Sol_{\varphi}(p'))$ and $p(\Sol_{\varphi}(p'))$ (this does not affect the asymptotic running time).
In the for loops we only update the current value of $\apx:=p(\APX)$ instead of updating $\APX$ explicitly each time. Furthermore for each tentative solution $\APX'$ we only compute $p(\APX')$ and $\sum_{i\in \APX':e\in P(i)}d(i)$ for each $e\in E$. This can be done in $O(|\Phi|m)$ time and it is sufficient to check whether $\APX'$ is a feasible solution and whether $p(\APX')>\apx$. We maintain the combination of the parameters $X^*_\varphi$ and $\tilde{\opt}^*$ that lead to the current value of $\apx$. 
At the end of the process from the optimal parameters $X^*_\varphi$ and $\tilde{\opt}^*$ we derive a corresponding solution $\APX$ of profit $\apx$ in $O(n+|\Phi|)=O(n^2)$ extra time. 

\begin{lemma}\label{lem:eptasUFP_feasibility}
\Cref{alg:UFP2} produces a feasible \UFP\ solution.
\end{lemma}
\begin{proof}
Obviously since $\APX=\emptyset$ is a feasible solution, and whenever we update $\APX$, we do that with the value $\APX'$ of a feasible solution.     
\end{proof}

\begin{lemma}\label{lem:eptasUFPapx}
\Cref{alg:UFP2} produces a $(1-2\eps)$-approximate solution.
\end{lemma}
\begin{proof}
Let us show that $p(\APX)\geq (1-\eps)p(\OPT')$. Notice that $\opt'=p(\OPT')\in [\frac{n^2}{\eps}]$, hence there is a value $\tilde{\opt}$ considered by the algorithm such that $\frac{1}{1+\eps}\opt'<\tilde{\opt}\leq \opt'$. Let us focus on execution of the external for loop with that value of $\tilde{\opt}$. 

Recall that $\opt'_\varphi=p(\OPT'_\varphi)=p(\OPT'\cap T_\varphi)$. As already argued before, there are corresponding values $(X_{\varphi})_{\varphi\in \Phi}$ considered by the algorithm such that $\tilde{\opt}_\varphi=X_{\varphi}\cdot \frac{\eps}{|\Phi|}\cdot \tilde{\opt}$ satisfies:
$$
\opt'_\varphi-\frac{\eps}{|\Phi|}\tilde{\opt}\leq \tilde{\opt}_\varphi\leq \opt'_\varphi.
$$
Let us focus on the execution of the inner for loop with these values of $X_{\varphi}$ (hence $\tilde{\opt}_\varphi$). The profit of the corresponding solution $\APX'$ is at least 
$$
\sum_{\varphi\in \Phi}\tilde{\opt}_\varphi\geq \sum_{\varphi\in \Phi}\opt'_\varphi -\eps \tilde{\opt}=\opt'-\eps \tilde{\opt}\geq (1-\eps)\opt'.
$$
Observe that $\OPT'_\varphi=\OPT'\cap T_\varphi$ is a valid solution for the Knapsack instance $K_\varphi$ with profit $\opt'_\varphi$, where $p_\varphi\geq \opt'_\varphi\geq \ceil{\tilde{\opt}_\varphi}$, hence also a valid candidate solution for $\Sol_\varphi\left({\min}\{p_\varphi,\ceil{\tilde{\opt}_\varphi}\right\})$. As a consequence $d(\Sol_\varphi({\min}\{p_\varphi,\lceil\tilde{\opt}_\varphi\rceil\})\leq d(\OPT'_\varphi)$. We conclude that $\APX'$ is a feasible solution. In more detail, for each $e\in E$,
$$
{\sum_{i\in \APX':e\in P(i)}d(i)}=
\sum_{\varphi\in \Phi: e\in \varphi}d(\Sol_\varphi(\ceil{\tilde{\opt}_\varphi}))\leq \sum_{\varphi\in \Phi: e\in \varphi}d(\OPT'_\varphi)={\sum_{i\in \OPT':e\in P(i)}d(i)}\leq u(e).
$$
It follows that $p(\APX)\geq p(\APX')\geq (1-\eps)\opt'$. Using standard arguments, we conclude that 
\begin{align*}
w(\APX) & \geq\, \frac{\eps \cdot w_{\max}}{n}\cdot p(\APX)\\
&\geq\, (1-\eps)\cdot \frac{\eps \cdot w_{\max}}{n}p(\OPT')\\
&\geq\, (1-\eps)\cdot \frac{\eps \cdot w_{\max}}{n}\cdot p(\OPT)\\
& \geq\, (1-\eps)\cdot \left(\frac{\eps \cdot w_{\max}}{n}\left(\frac{n}{\eps \cdot w_{\max}}\cdot w(\OPT)-n\right)\right)\\
&=\,(1-\eps)\cdot (\opt-\eps\cdot  w_{\max})\geq (1-\eps)\cdot (1-\eps)\opt.
\end{align*}

\end{proof}


\newcommand{\iters}{\textnormal{\texttt{iters}}}
It remains to upper bound the running time. Let $\iters$ be the number of iterations of the inner loop in \Cref{alg:UFP2}, i.e. the number of possible valid combinations for $(X_\varphi)_{\varphi\in \Phi}$. The bound on the running time follows easily from the following technical lemma. 
	\begin{lemma}\label{lem:Lbound}
		$\iters \leq \left(\frac{1+2\eps}{\eps}\cdot e\right)^{\abs{\Phi}} $.
	\end{lemma}
 
\begin{lemma}\label{lem:eptasUFPtime}
\Cref{alg:UFP2} runs in time $O\left(\frac{n^3}{\eps}+\left(\frac{1}{\eps}\right)^{O(m^2)}\cdot m^3\cdot \log n\right)$.
\end{lemma}
\begin{proof}
We already argued that the dynamic tables can be computed in total time $O(\frac{n^3}{\eps})$. We also observed that the outer for loop is executed   $O\left(\frac{1}{\eps}\log n\right)$ times. As already discussed, lines \ref{step:setTildeOpt}-\ref{step:update2} take $O\left(|\Phi|\cdot m\right)$ time. Thus the second phase of the algorithm can be implemented in time $O(n^2+|\Phi|\cdot m\cdot \iters\cdot \frac{1}{\eps}\cdot \log n)$ time. By Lemma \ref{lem:Lbound}, the overall running time of the algorithm is
$$
O\left(\frac{n^3}{\eps}+\left(\frac{1+2\eps}{\eps}\cdot e\right)^{\abs{\Phi}}\cdot m\cdot|\Phi|\cdot  \frac{1}{\eps}\log n\right).
$$
The claim follows since $|\Phi|\leq \frac{1}{2}\cdot m\cdot (m+1)$.
\end{proof}

 It remains to prove \Cref{lem:Lbound}. To that aim, we need a standard bound on the binomial coefficients.  Let $\entropy(x) =-x\cdot  \ln( x) -(1-x)\cdot \ln(1-x)$ be the entropy function and assume $\entropy(0)=\entropy(1)=0$. 
\begin{lemma}[Example 11.1.3 in \cite{CoverT2006}]
	\label{lem:binom}
For every $n\in \mathbb{N}$ and integer $0\leq k\leq n$ it holds that $\binom{n}{k} \leq \exp\left(n \cdot \entropy\left(\frac{k}{n}\right)\right)$. 
\end{lemma}
 
 \begin{proof}[Proof of \Cref{lem:Lbound}]
Recall that $\iters$ is equal to the possible sequences of $|\Phi|$ non-negative integers whose sume is at most $Y=\floor{ \frac{1+\eps}{\eps}\cdot |\Phi|}$. These sequences can be represented via a binary string as follows. Let $\varphi_1,\ldots\varphi_{|\Phi|}$ be an arbitrary ordering of $\Phi$, and $X_i=X_{\varphi_i}$. The bit string consists of $X_1$ many $1$s, followed by one $0$, followed by $X_2$ many $1$s and so on, ending with the  $X_{|\Phi|}$ many $1$s, an additional $0$ and a final padding of $1$s till the target length of $Y$ is reached. In particular all valid sequences correspond to binary strings with $Y+|\Phi|$ 
digits and exactly $|\Phi|$  zeros. It is therefore sufficient to {upper bound the number of the latter bit strings, namely ${{Y+|\Phi}\choose{|\Phi|}}$.}
By \Cref{lem:binom} we have,   
		\begin{equation}
			\label{eq:Lbound_second2}
			\begin{aligned}
		\iters\,
		&{=}\, \binom{\abs{\Phi}+ \floor{\abs{\Phi}\cdot  \frac{1+\eps}{\eps}}}{\abs{\Phi}}\\ 
		\,&\leq\, \exp\left(  \left( \abs{\Phi}+ \floor{\abs{\Phi}\cdot \frac{1+\eps}{\eps}}\right)\cdot \entropy\left(\frac{\abs{\Phi}}{\abs{\Phi}+ \floor{\abs{\Phi}\cdot  \frac{1+\eps}{\eps}}}\right) \right) 
		\\ 
		\,&\leq \, \exp\left(  \left( {\abs{\Phi}\cdot \frac{1+2\eps}{\eps}}\right)\cdot \entropy\left(\frac{\abs{\Phi}}{{\abs{\Phi}\cdot \frac{1+2\eps}{\eps}}}\right) \right)=\, \left(\exp\left( \frac{1+2\eps}{\eps}\cdot \entropy\left(\frac{\eps}{1+2\eps}\right)\right)\right)^{\abs{\Phi}},
		\end{aligned}
		\end{equation}
		where the last inequality holds since $x\cdot\entropy\left( \frac{a}{x}\right)$ is increasing in $x$ for any $a\geq 1$ and $\abs{\Phi}+ \floor{\abs{\Phi}\cdot \frac{1+\eps}{\eps}}\leq \abs{\Phi}\cdot \frac{1+2\eps}{\eps}$.   
It also holds that 
		\begin{equation}
			\label{eq:entropy_exp2}
			\begin{aligned}
				\frac{1+2\eps}{\eps}\cdot \entropy\left(\frac{\eps}{1+2\eps}\right)\,&=\,
				 \frac{1+2\eps}{\eps}\cdot\left( -  \frac{\eps}{1+2\eps} \cdot\ln\left( \frac{\eps}{1+2\eps}\right) - \left(1-\frac{\eps}{1+2\eps}\right)\cdot  \ln \left(1-\frac{\eps}{1+2\eps}\right)\right) \\
				 &\leq\, - \ln \left(\frac{\eps}{1+2\eps}\right)  -\frac{1+2\eps}{\eps}\cdot \left(1-\frac{\eps}{1+2\eps}\right) \cdot\left( - \frac{\eps}{1+2\eps}  \left(1+ \frac{\eps}{1+2\eps} \right) \right)\\
				 &=\,  \ln \left(\frac{1+2\eps}{\eps}\right)  +\left(1- \frac{\eps}{1+2\eps} \right)\left(1+\frac{\eps}{1+2\eps}\right)\\
				 &\leq\ln \left(\frac{1+2\eps}{\eps}\right)  +1
				\end{aligned}
		\end{equation}
		where the first inequality follows from $\ln(1-x)\geq -x(1+x)$ for $x\in(0.0.1)$, and the second inequality holds as $(1+x)(1-x)\leq 1$ for every $x\in \mathbb{R}$.  By \eqref{eq:Lbound_second2} and \eqref{eq:entropy_exp2} we have,
		$$
			\iters\,\leq \,  \left(\exp\left( \frac{1+2\eps}{\eps}\cdot \entropy\left(\frac{\eps}{1+2\eps}\right)\right)\right)^{\abs{\Phi}}
			\,\leq\,  \left(\exp\left( \ln\left(\frac{1+2\eps}{\eps}\right)+1\right)\right)^{\abs{\Phi}} 
			\,= \, \left(\frac{1+2\eps}{\eps}\cdot {e}\right)^{\abs{\Phi}} 
				$$
	\end{proof}

 We now have the tools required to  complete the proof of \Cref{thm:eptasUFP}. 
\begin{proof}[Proof of \Cref{thm:eptasUFP}]
It follows directly from Lemmas \ref{lem:eptasUFP_feasibility}, \ref{lem:eptasUFPapx} and \ref{lem:eptasUFPtime} by choosing the parameter $\eps/2$ so as to have a $(1-\eps)$ approximation. 
\end{proof}

\section{A Lower bound for \BagUFP}
\label{sec:BagUFPhardness}
In this section we prove \Cref{thm:BagUFPnoFPTAS} using a simple reduction from the partition problem.

\subsubsection*{Proof of \Cref{thm:BagUFPnoFPTAS}:}

Recall that in the $NP$-complete Partition problem we are given a collection of $n$ non-negative integers $A=\{a_1,\ldots,a_n\}$ in $[0,1]$ whose sum is $2M$. Our goal is to determine whether there exists a subset of numbers whose sum is precisely $M$. 

We show that an \FPTAS\ for \BagUFP\ in the considered case implies a polynomial time algorithm to solve Partition, hence the claim. We build (in polynomial time) an instance of \BagUFP\ with 2 edges $e_1$ and $e_2$, both of capacity $M$. Furthermore, for each $a_j$, we create two tasks $t^1_j$ and $t^2_j$, with demand $a_j$ and subpath $e_1$ and $e_2$, resp. All the tasks have profit $1$. The bags are given by the pairs $\{t^1_j,t^2_j\}$, $j=1,\ldots,n$. Obviously, the input Partition instance is a YES instance iff the optimal solution to the corresponding \BagUFP\ instance has value $n$, i.e. exactly one task per bag is selected (notice that a solution cannot have larger profit). Indeed, given a solution $A'\subseteq A$ for the Partition instance, a valid solution to the corresponding \BagUFP\ instance is obtained by selecting all the tasks $t^1_j$ with $j\in A'$ and all the tasks $t^2_j$ with $j\notin A'$. Notice that the total demand of the tasks using $e_1$ and $e_2$ must be exactly $M$. Vice versa, given a \BagUFP\ solution $S$ of profit $n$, the selected tasks $S_1\subseteq S$ of type $t^1_j$ must have total demand exactly $M$, hence inducing a valid Partition solution $A':=\{j\in \{1,\ldots,n\}:t^1_j\in S_1\}$.

We run the mentioned \FPTAS\ on the obtained \BagUFP\ instance with parameter $\eps=\frac{1}{2n}$ (hence taking polynomial time). If the optimal solution is $n$, the \FPTAS\ will return a solution of profit at least $\frac{n}{1+\eps}\geq n-\frac{1}{2+1/n}>n-1$, hence a solution of profit $n$ since the profit is an integer. Otherwise, the \FPTAS\ will return a solution of profit at most $n-1$. This is sufficient to discriminate between YES and NO instances of Partition. \qed
\section{A Lower bound for UFP}
\label{sec:UFPhardness}
In this section we prove \Cref{thm:hardnessUFP}.
We give a reduction from the following multiple-choice variant of the $k$-{\em subset sum} problem.   

\begin{definition}
	\label{def:SSM}
	{\bf Subset Sum with Multiple-Choice (SSM):}  For some $n \in \mathbb{N}$ and $k \in [n]$, let $A_1 = \{a^1_1,\ldots,a^1_n\}, A_2 = \{a^2_1,\ldots,a^2_n\}, \ldots, A_k = \{a^k_1,\ldots,a^k_n\} \subset \mathbb{R}_{> 0}$ be $k$ sets of $n$ numbers, and let $B \in \mathbb{R}_{> 0}$ be a target value. A {\em solution} for the instance is indices $r(i) \in [n]$, for all $i \in [k]$, such that $\sum_{i \in [k]} a^i_{r(i)} = B$. The goal is to decide if there is a solution.  The {\em parameter} of the instance is $k$. 
\end{definition}

We prove the following hardness result for SSM. Even though the proof does not contain new ideas, the result can be of an independent interest. 

\begin{lemma}
	\label{lem:SSM}
	Unless \textnormal{W[1] = FPT}, for every function $f:\mathbb{N} \rightarrow \mathbb{N}$ there is no algorithm that decides \textnormal{SSM} in time $f(k) \cdot n^{O(1)}$, where $n$ is the encoding size of the instance and $k$ is the parameter. 
\end{lemma} We prove \Cref{lem:SSM} in \Cref{sec:SSM}.  
In the following, we give a reduction from SSM to UFP. After the formal construction, we provide some intuition. 
\subsection{Reduction from SSM to UFP with $m = O(k)$ edges}
For the remaining of this section, let  $A_1 = \{a^1_1,\ldots,a^1_n\}, A_2 = \{a^2_1,\ldots,a^2_n\}, \ldots, A_k = \{a^k_1,\ldots,a^k_n\} \subset \mathbb{R}_{> 0}$, and $B \in \mathbb{R}_{> 0}$ be an SSM instance $S$.\footnote{For encoding reasons, assume without the loss of generality that all numbers are rational numbers.} Let $$r(S) =2 \cdot B+\sum_{i \in [k]} \sum_{j \in [n]} a^i_j $$ be the sum of all numbers in the instance with an additive factor of $2 \cdot B$. Without the loss of generality, we may assume that $2\cdot B<1$ and that $0 < B<r(S) < 1$; in addition, by a simple reduction, we may assume that for all $i \in [k], j \in [n]$ it holds that  $a^i_j < \frac{2 \cdot B}{k}$ (such assumptions are frequently used for subset sum problems, e.g., \cite{W17}). Given the instance $S$, in our reduction we create $O( n\cdot k)$ UFP instances, all of them with $O(k)$ edges and $O(n \cdot k)$ tasks. Specifically, for each integer $k \leq x \leq k \cdot n$ 
we define the UFP instance $U_{x}(S)$ as follows. 
%
%
We use the following auxiliary functions in the construction,
\begin{equation}
	\label{eq:W}
	\begin{aligned}
		W(S) ={} & 1+r(S)\\
		Q(S) ={} &  1+(k \cdot n+1) \cdot W(S).
	\end{aligned}
\end{equation}

{\bf The path.} Define a path with the vertices $v'_0, v_1, v'_1,v_2,v'_2 \ldots, v_k,v_k',v_{k+1}$. Denote the leftmost and rightmost edges by $e_0, e_k$, respectively. In addition, for every $i \in [k]$ let the edge between $v_i$ and $v'_i$ be $f_i = (v_i,v'_i)$, and let $e_h = (v'_h,v_{h+1})$ for $h \in [k] \cup \{0\}$ in general.  
{ Define the capacity of the first and last edge by}
$$u_{x,S}(e_0) = u_{x,S}(e_k) =k \cdot Q(S)+ x \cdot W(S)+B.$$

Define the capacity of any other edge $e_h, h \in [k-1]$ and $f_i, i \in [k]$ by $$u_{x,S}(e_h) = u_{x,S}(f_i) = k \cdot Q(S)+x \cdot W(S)+r(S).$$ 

\begin{figure}
		\centering
		\begin{tikzpicture}[scale=1.4, every node/.style={draw, circle, inner sep=1pt}]
			\node (v0') at (-1.5,0) [circle, minimum size=0.9cm]  {$v'_0$};
			\node (v1) at (0,0) [circle, minimum size=0.9cm]  {$v^{    }_1$};
			\node (v1') at (1.5,0) [circle, minimum size=0.9cm]  {$v'_1$};
			\node (v2) at (3,0) [circle, minimum size=0.9cm]  {$v^{    }_2$};
			\node (v2') at (4.5,0) [circle, minimum size=0.9cm]  {$v'_2$};
			\node[draw=none] at (5.5, 0) {${\bf \cdots}$}; 
			\node (vk) at (6.5,0) [circle, minimum size=0.9cm] {$v^{    }_k$};
			\node (vk') at (8,0)  [circle, minimum size=0.9cm]  {$v'_k$};
			\node (vk+1) at (9.5,0) {$v^{    }_{k+1}$};

			\node[draw=none] (delta2) at (8.75, -0.55) {$ \textcolor{red}{q^{k}_j}$};
			
			\node[draw=none] (delta2) at (-0.75, 0.55) {$ \textcolor{blue}{z^{1}_j}$};
			
			\node[draw=none] (delta2) at (2.3, -0.6) {$ \textcolor{blue}{z^{2}_j}$};
			
			\node[draw=none] (delta2) at (5.4, 0.75) {$ \textcolor{red}{q^{2}_j}$};
			
			\node[draw=none] (delta2) at (4.5, 1.45) {$ \textcolor{red}{q^{1}_j}$};
			
			\node[draw=none] (delta2) at (4.5, -1.25) {$ \textcolor{blue}{z^{k}_j}$};
			
			\draw[bend left,dashed,line width=1pt, color=blue] (v0') to (v1);
			\draw[bend right,dashed,line width=2pt, color=blue] (v0') to (v2);
			\draw[bend right,dashed,line width=3pt, color=blue] (v0') to (vk);
			
			\draw[bend right,dashed,line width=1pt, color=red] (vk') to (vk+1);
			\draw[bend left,dashed,line width=2pt, color=red] (v2') to (vk+1);
			\draw[bend left,dashed,line width=3pt, color=red] (v1') to (vk+1);
			\draw (v0') -- (v1) node[midway, below, inner sep=-1pt, draw=none] {$e_0$}; 
			\draw (v1) -- (v1') node[midway, below, inner sep=-1pt, draw=none] {$f_1$};
			\draw (v2) -- (v2') node[midway, below, inner sep=-1pt, draw=none] {$f_2$};
			\draw (v1') -- (v2) node[midway, below, inner sep=-1pt, draw=none] {$e_1$};
			\draw (vk) -- (vk') node[midway, below, inner sep=-1pt, draw=none] {$f_k$};
			\draw (vk') -- (vk+1) node[midway, below, inner sep=-1pt, draw=none] {$e_k$};

			\draw[bend left,dashed,line width=1pt, color=brown] (v1) to (v1');
			\node[draw=none] (delta1) at (0.75, 0.5) {$\textcolor{brown}{\delta_1} $};

			\draw[bend left,dashed,line width=1pt, color=brown] (v2) to (v2');
			\node[draw=none] (delta2) at (3.75, 0.5) {$\textcolor{brown}{\delta_2} $};
			
			\draw[bend left,dashed,line width=1pt, color=brown] (vk) to (vk');
			\node[draw=none] (delta2) at (7.25, 0.5) {$\textcolor{brown}{\delta_k} $};
		\end{tikzpicture}
		\caption{\label{fig:Construction} An illustration of the construction. The path $v'_0,v_1,v'_1,\ldots, v_k,v'_k, v_{k+1}$ of the interleaving sequences of edges $e_0,f_1,e_1,f_2, \ldots, f_k,e_k$ is shown along with the subpaths of the tasks $\delta_1,\ldots, \delta_k$ (in dashed brown), the subpaths of the tasks $z^1_j,\ldots, z^k_j$ (in dashed blue, for some $j$), and the subpaths of the tasks $q^1_j,\ldots, q^k_j$ (in dashed red). Tasks with larger weights have thicker lines for their subpaths. }
\end{figure}
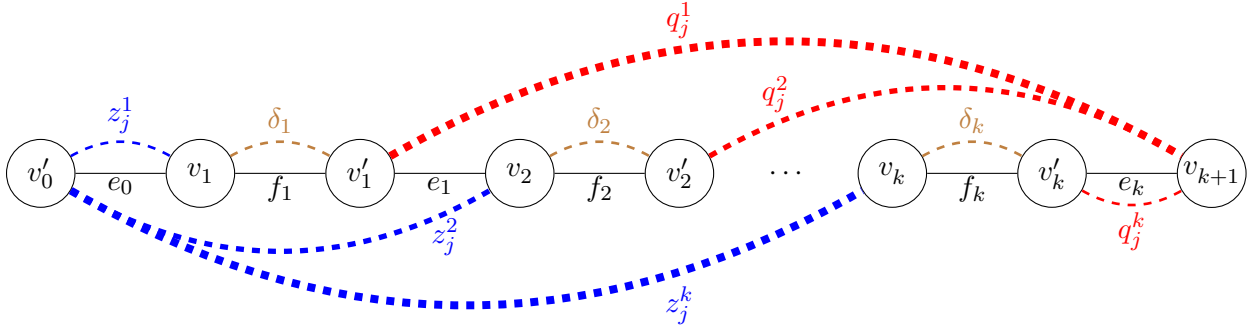

\noindent{\bf The tasks.} We define three sets of tasks. For the first two sets of tasks, for every $i \in [k]$ we define the tasks $z^i_1, z^i_2, \ldots, z^i_n$ and $q^i_1, q^i_2, \ldots, q^i_n$ such that for all $j \in [n]$ the demands of $z^i_j$ and $q^i_j$ are 
\begin{equation*}
	\begin{aligned}
		d_{S}\left(z^i_j\right) ={} & Q(S)+ j \cdot W(S)+a^i_j\\
		d_{S}\left(q^i_j\right) ={} & Q(S)+ j \cdot W(S)+\frac{2\cdot B}{k}-a^i_j.
	\end{aligned}
\end{equation*}
The subpaths of $z^i_j$ and $q^i_j$ are described as $s(z^i_j) = v'_0, t(z^i_j) = v_i$, and $s(q^i_j) = v'_i, t(q^i_j) = v_{k+1}$. Define the auxiliary parameters 
\begin{equation}
	\label{eq:H}
	\begin{aligned}
		L(S) ={} & 2 \cdot (k \cdot n)^2+1\\
		H(S)={} & 2k^2 \cdot L(S)+2 \cdot k^2 \cdot n+1. 
	\end{aligned}
\end{equation}
Finally, define the profits (weights) of the tasks as 
\begin{equation*}
	\begin{aligned}
		w_{S}(z^i_j) ={} & H(S)+i \cdot L(S)+ j \cdot i\\
		w_{S}(q^i_j) ={} & H(S)+(k+1-i) \cdot L(S)+(k+1-i) \cdot j.
	\end{aligned}
\end{equation*} 

The last set of tasks is $\delta_1,\ldots,\delta_k$. For each $i \in [k]$, the subpath of $\delta_i$ is  $s(\delta_i) = v_i$, $t(\delta_i) = v'_i$. Moreover, the demand and profit of $\delta_i$ are $d_{S}(\delta_i) = Q(S), w_S(\delta_i) = H(S)$, respectively. When clear from the context, we omit the reference to $S$ or to some $k \leq x \leq k \cdot n$ from the parameters as follows: $$U_x = U_{x}(S), W = W(S),Q = Q(S), u_{x,S} = u,d = d_S,L = L(S), H = H(S),w = w_S.$$
We give an illustration of the construction in \Cref{fig:Construction}. 

\noindent{\bf Intuition.} 
Each pair of tasks $z^i_j,q^i_j$, where $i \in [k],j \in [n]$, represents the number $a^i_j$ in the SSM instance $S$. From a solution  $r(1),\ldots,r(k)$ to $S$ we can construct a solution for the reduced instance $U_x$, where $x = \sum_{i \in [k]} r(i)$ is the sum of indices of the solution for $S$; the construction of the solution for $U_x$ takes all corresponding tasks $z^i_{r(i)},q^i_{r(i)},i \in [k]$ and $\delta_1,\ldots,\delta_k$. This gives the easy direction of the reduction. The following paragraphs gives intuition for the (more difficult) other direction of the reduction.

The parameter $Q$ and the capacities of the edges are defined such that the number of tasks intersecting each edge is at most $k$. Additionally, $W$ is defined such that the sum of indices $j$ of the tasks cannot exceed the parameter $x$ in any of the edges.  On the other hand, the parameter $H$ is sufficiently large such that the profit of a solution with a  larger number of tasks is always larger than the profit of a solution with a smaller number of tasks. To a smaller degree, the profit of the tasks $z^i_j,q^i_j$ is proportional to their length ($i$ and $k+1-i$, respectively), multiplied by the adjusting parameter $L$; as we know that the number of tasks intersecting any edge is bounded by $k$, an optimal solution must take tasks that have a total length large as possible (and of course to take as many tasks as possible). This restricts very profitable solutions to take the same number of tasks from $z^i_1,\ldots,z^i_n$ and $q^i_1,\ldots,q^i_n$. 

To a third degree (the parameter $L$ is large enough), the profit of the tasks  $z^i_j,q^i_j$ is also proportional to the value $j$ multiplied by the length of the tasks; this gives a motivation for the optimum (for the value of $x$ corresponding to the sum of indices in a solution for $S$) to take tasks with larger values of $j$, but it cannot get too high as the demand is also proportional to the $j$-value of the tasks. 
Finally, note that the tasks $\delta_i, i \in [k]$ are short tasks of length $1$; their path does not intersect the paths of $z^i_j,q^i_j$; therefore, to gain an optimal profit the optimum selects all tasks $\delta_1,\ldots,\delta_k$. In addition, the optimum (if it has a high enough profit) chooses exactly
one pair of tasks $z^i_{r(i)},q^i_{r(i)}$ for each index $i \in [k]$. By the capacity constraint of the first and last edges $e_0,e_k$ this yields a solution $r(1),\ldots,r(k)$ for the SSM instance $S$.  
This intuition is summarized in the next result (we give the proof at the end of this section). 
\begin{lemma}
 \label{thm:aux}
	There is a solution for $S$ if and only if there is 
	$k \leq x \leq k \cdot n$ such that there is a solution for $U_{x}$ of profit at least $3k \cdot H+k \cdot (k+1) \cdot L+k \cdot (k+1) \cdot x$. 
\end{lemma}
Using \Cref{thm:aux}, we can show that a \pF~for UFP can be used to decide SSM in FPT time. 

\begin{lemma}
	\label{lem:alg}
	Let $\cA$ be a \pF~for \textnormal{UFP} parametrized by length of path. 
	Then, there is an algorithm that decides an \textnormal{SSM} instance $S$ with parameter $k$ in time 
	$g(k) \cdot |S|^{O(1)}$, where $g$ is some computable function and $|S|$  is the encoding size of $S$.
\end{lemma}

\begin{proof}
	Assume that there is a \pF~$\cA$ for UFP parametrized by the length of path. Thus, there is a computable function $f$ such that for every UFP instance $I$ with $m$ edges, the running time of $\cA$ on $I$ is bounded by $ {f(m) \cdot \left(\frac{|I|}{\eps}\right)}^{O(1)}$.
 Let $A_1 = \{a^1_1,\ldots,a^1_n\}, A_2 = \{a^2_1,\ldots,a^2_n\}, \ldots, A_k = \{a^k_1,\ldots,a^k_n\} \subset \mathbb{R}_{> 0}$, and $B \in \mathbb{R}_{> 0}$ be an SSM instance $S$. Define the following algorithm $\cB$ that decides $S$ using $\cA$.
	
	\begin{enumerate}
		\item  Define error parameter $\eps = \frac{1}{20 \cdot H \cdot n \cdot k}$ (recall the definition of $H = H(S)$ in \eqref{eq:H}). 
		\item For all $k \leq x \leq k \cdot n$: Execute $\cA$ on $U_x(S)$  and obtain a solution $R_x$. 
		\item Return that $S$ is a yes-instance if and only if there is $k \leq x \leq k \cdot n$ with profit $$w_S(R_x) \geq 3k \cdot H+k \cdot (k+1) \cdot L+k \cdot (k+1) \cdot x.$$ 
	\end{enumerate} Observe that the length of the path of the instance $U_x(S)$ for all $k \leq x \leq k \cdot n$ is $2k+1 = O(k)$; moreover, the encoding size of $U_x(S)$ is bounded by $|S|^{O(1)}$. Thus, by the running time guarantee of $\cA$, the running time of $\cB$ on $S$ is $f(k) \cdot |S|^{O(1)}$. 
	Moreover,  for all $k \leq x \leq k \cdot n$ it holds that 
	\begin{equation}
		\label{eq:t=r2}
		\begin{aligned}
			w_S(R_x) \geq{} & (1-\eps) \cdot \opt(U_x) \\
   ={} &  \opt(U_x)- \frac{\opt(U_x)}{20 \cdot H \cdot n \cdot k} \\
   \geq{} & \opt(U_x)- \frac{\opt(U_x)}{2 \cdot \opt(U_x)} \\
   ={} & \opt(U_x)-\frac{1}{2}. 
		\end{aligned}
	\end{equation}
	The second inequality holds since $20 \cdot H \cdot n \cdot k$ is a strict upper bound on $2 \cdot \opt(U_x)$: observe that the profit of each task in $U_x$ is strictly less than $3 \cdot H$ by \eqref{eq:H}; therefore, by taking all $2 \cdot n \cdot k+k$ tasks 
 in $U_x$ to the solution, we get 
 $$\opt(U_x) < 3 \cdot H \cdot \left(2 \cdot n \cdot k+k \right) \leq 3 \cdot H \cdot \left(3 \cdot n \cdot k \right)< 10 \cdot H \cdot n \cdot k.$$  Therefore, since all profits of tasks are integers, by \eqref{eq:t=r2}  for all $k \leq x \leq k \cdot n$ it holds that $w_S(R_x) = \opt(U_x(S))$. Hence, 
	by \Cref{thm:aux} we have that $\cB$ decides the SSM instance $S$ correctly. 
\end{proof}
We can finally give our main result as a consequence of the above.

\noindent {\bf Proof of  \Cref{thm:hardnessUFP}:} The proof follows by \Cref{lem:SSM} and \Cref{lem:alg}. \qed

In the remaining of this section, we prove \Cref{thm:aux}. 
We start by proving the easier direction - constructing a solution for some $U_x$ based on a solution for $S$. 

\begin{lemma}
	\label{lem:if}
	If there is a solution for $S$ then there is 
	$k \leq x \leq k \cdot n$ such that there is a solution for $U_{x}$ of profit at least $3k \cdot H+k \cdot (k+1) \cdot L+k \cdot (k+1) \cdot x$. 
\end{lemma}

\begin{proof}
	Let $r(1), r(2),\ldots, r(k)$ be a solution for $S$; that is, it holds that $\sum_{i \in [k]} a_{r(i)} = B$. Define $$x = \sum_{i \in [k]} r(i).$$
	Since $r(i) \in [n]$, it follows that $k \leq x \leq k \cdot n$. We construct a solution $F$ for $U_{x}$ of profit $$3k \cdot H+k \cdot (k+1) \cdot L+k \cdot (k+1) \cdot x.$$ 
	Define $F_z = \{z^i_{r(i)}~|~i \in [k]\}$, $F_q =  \{q^i_{r(i)}~|~i \in [k]\}, F_{\delta} = \{\delta_i~|~i \in [k]\}$, and define the solution of $U_x$ as $$F = F_z \cup F_q \cup F_{\delta}.$$ 
	
	We now show that $F$ does not violates the capacities of the edges. For $e_0$:

	\begin{equation}
		\label{eq:z}
		\begin{aligned}
			\sum_{t \in F, e_0 \in P(t)} d(t) ={} &  	\sum_{t \in F_z, e_0 \in P(t)} d(t)\\
			={} & \sum_{i \in [k]} d(z^i_{r(i)}) \\
			={} & \sum_{i \in [k]} \left( Q+ r(i) \cdot W+a^i_{r(i)} \right)\\
			={} &  k \cdot Q+\sum_{i \in [k]} r(i) \cdot W+\sum_{i \in [k]} a^i_{r(i)}\\
			={} &  k \cdot Q+x \cdot W+B\\
			={} &  u(e_0).\\
		\end{aligned}
	\end{equation}
	Now, for $e_k$: 
	\begin{equation}
		\label{eq:q}
		\begin{aligned}
			\sum_{t \in F, e_k \in P(t)} d(t) ={} &  	\sum_{t \in F_q, e_k \in P(t)} d(t)\\
			={} & \sum_{i \in [k]} d(q^i_{r(i)}) \\
			={} & \sum_{i \in [k]} \left(Q+r(i) \cdot W+\frac{2\cdot B}{k}-a^i_{r(i)}\right)\\
			={} &  k \cdot Q+\sum_{i \in [k]} r(i) \cdot W+\sum_{i \in [k]} \left(\frac{2\cdot B}{k}-a^i_{r(i)}\right)\\
			={} &  k \cdot Q+x \cdot W +2 \cdot B -\sum_{i \in [k]} a^i_{r(i)}\\
			={} &  k \cdot Q+x \cdot W +2 \cdot B-B\\
			={} &  k \cdot Q+x \cdot W +B\\
			={} &  u(e_k).\\
		\end{aligned}
	\end{equation} For any $i \in [k]$ the solution $F$ satisfies the capacity constraint of $e_i$: 
	\begin{equation}
		\label{eq:e}
		\begin{aligned}
			\sum_{t \in F, e_i \in P(t)} d(t) ={} & \sum_{i' \in [k], i' > i} d(z^{i'}_{r(i')})+ \sum_{i' \in [k], i' \leq i} d(q^{i'}_{r(i')}) \\
			={} & \sum_{i' \in [k], i' > i} \left( Q+ r(i') \cdot W+a^{i'}_{r(i')} \right)+ \sum_{i' \in [k], i' \leq i} \left(Q+r(i') \cdot W+\frac{2\cdot B}{k}-a^{i'}_{r(i')}\right)\\ 
			\leq{} & k \cdot Q+W \cdot \sum_{i' \in [k]} r(i')+\sum_{i' \in [k]} \max \left\{ \frac{2 \cdot B}{k}-a^{i'}_{r(i')},a^{i'}_{r(i')} \right\}\\
			\leq{} &   k \cdot Q+x \cdot W+2 \cdot B+\sum_{i' \in [k]} a^{i'}_{r(i')} \\
			\leq{} &   k \cdot Q+x \cdot W+r(S) \\
			={} & u(e_i). 
		\end{aligned}
	\end{equation} 
	
	
	Finally, for any $i \in [k]$ 
 the solution $F$ satisfies the capacity constraint of $f_i$:
	\begin{equation}
		\label{eq:f}
		\begin{aligned}
			\sum_{t \in F, f_i \in P(t)} d(t) ={} & d(\delta_i)+\sum_{i' \in [k], i' > i} d(z^{i'}_{r(i')})+ \sum_{i' \in [k], i' < i} d(q^{i'}_{r(i')}) \\
			={} & Q+\sum_{i' \in [k], i' > i} \left( Q+ r(i') \cdot W+a^{i'}_{r(i')} \right)+ \sum_{i' \in [k], i' < i} \left(Q+r(i') \cdot W+\frac{2\cdot B}{k}-a^{i'}_{r(i')}\right)\\ 
			\leq{} & k \cdot Q+W \cdot \sum_{i' \in [k]} r(i')+\sum_{i' \in [k]} \max \left\{ \frac{2 \cdot B}{k}-a^{i'}_{r(i')},a^{i'}_{r(i')} \right\}\\
			\leq{} &   k \cdot Q+x \cdot W+2 \cdot B+\sum_{i' \in [k]} a^{i'}_{r(i')} \\
			\leq{} &   k \cdot Q+x \cdot W+r(S) \\
			={} & u(f_i). 
		\end{aligned}
	\end{equation} 
	
	Thus, by \eqref{eq:z}, \eqref{eq:q}, \eqref{eq:e}, and \eqref{eq:f} it holds that $F$ is a feasible solution for $U_{x}$.  To conclude, we show that the total profit of $F$ is $w(F) = 3k \cdot H+k \cdot (k+1) \cdot L+k \cdot (k+1) \cdot x$. 
	
	\begin{equation*}
		\label{eq:3}
		\begin{aligned}
			w(F) ={} & \sum_{t \in F} w(t) \\
			={} & \sum_{t \in F_z} w(t)+\sum_{t \in F_q} w(t)+\sum_{t \in F_{\delta}} w(t)\\
			={} & \sum_{i \in [k]} w\left(z^i_{r(i)}\right)+\sum_{i \in [k]} w\left(q^i_{r(i)}\right)+\sum_{i \in [k]} w(\delta_i)\\
			={} & \sum_{i \in [k]} \bigg( H+L \cdot i+ i \cdot r(i) \bigg) +\sum_{i \in [k]} \bigg( H+(k+1-i) \cdot L+(k+1-i) \cdot r(i)\bigg)+ \sum_{i \in [k]} H\\
			={} & 3k \cdot H+\sum_{i \in [k]} \bigg( i \cdot L+(k+1-i) \cdot L+i \cdot r(i)+ (k+1-i) \cdot r(i) \bigg)\\
			={} & 3k \cdot H+k \cdot (k+1) \cdot L+k \cdot (k+1) \cdot \sum_{i \in [k]} r(i)\\
			={} & 3k \cdot H+k \cdot(k+1) \cdot L+k \cdot (k+1) \cdot x.\\
		\end{aligned}
	\end{equation*}	
\end{proof}

We now prove the second direction; that is, we show that if there is a solution $F$ for $U_x$, for some 	$k \leq x \leq k \cdot n$, of profit at least $3k \cdot H+k \cdot (k+1) \cdot L+k \cdot (k+1) \cdot x$, then there are values $r(i) \in [n]$ for all $i \in [k]$ such that $F$ must be of the form $F = \left\{z^i_{r(i)}~|~i \in [k]\right\} \cup \left\{q^i_{r(i)}~|~i \in [k] \right\} \cup \{m_1,\ldots,m_k\}$. 
Consequently, we can show that $r(1),\ldots,r(k)$ is a solution for $S$. 
\begin{lemma}
	\label{lem:only}
	If there is 
	$k \leq x \leq k \cdot n$ such that there is a solution for $U_{x}$ of profit at least $3k \cdot H+k \cdot (k+1) \cdot L+k \cdot (k+1) \cdot x$, then there is a solution for $S$. 
\end{lemma}


we now prove \Cref{lem:only}. Before we prove the lemma, we give several useful properties that discover the structure of a solution of high profit. 
For this section, fix $k \leq x \leq k \cdot n$ and let $F$ be a solution for $U_{x}$ of profit at least $$3k \cdot H+k \cdot (k+1) \cdot L+k \cdot (k+1) \cdot x.$$
For every $e \in E$, let $F_e = \{t \in F~|~e \in P(t)\}$ be the set of tasks in $F$ intersecting $e$. The first elementary property of $F$ (and in fact, for any solution) that we show is that there cannot be more than $k$ tasks intersecting each edge. 
\begin{lemma}
	\label{claim:cardinality}
	for all $e \in E$ it holds that $|F_e| \leq k$.  
\end{lemma}
\begin{proof}

	%
	%
	Assume towards a contradiction that there is $e \in E$ such that $|F_e|>k$. Then, 
	\begin{equation}
		\label{eq:c1}
		\begin{aligned}
			\sum_{t \in F, e \in P(t)} d(t) \geq{} & |F_e| \cdot Q \\
			\geq{} & (k+1) \cdot Q \\
			>{} & k \cdot Q+k \cdot n \cdot W+W\\
			\geq{} & k \cdot Q+x \cdot W+W\\
			>{} & k \cdot Q+x \cdot W+r(S)\\
			\geq {} & u(e). 
		\end{aligned}
	\end{equation} 
	
		The first inequality holds since the demand of each task is at least $Q$. The second inequality relies on the assumption that $|F_e|>k$. The remaining inequalities use \eqref{eq:W}. By \eqref{eq:c1} we reach a contradiction to the feasibility of the solution $F$. 
	\end{proof}
	
	Consider the following partition of $F$: 
	\begin{equation*}
		\label{eq:111}
		\begin{aligned}
			F_z ={} & \{z^i_{j} \in F~|~i \in [k], j\in [n]\}\\
			F_q ={} &   \{q^i_{j} \in F~|~i \in [k], j\in [n]\}\\
			F_{\delta} ={} & \{\delta_i \in F~|~i \in [k]\}.
		\end{aligned}
	\end{equation*}
	These are the tasks in $F$ partitioned by their {\em type}; note that $F_z=F_{e_0}$ and  $F_q = F_{e_k}$.  The second property that we prove, relying on \Cref{claim:cardinality} and the high profit of $F$, is that the number of tasks in $F$ from each type is exactly $k$. 
	\begin{lemma}
		\label{claim:cardinality2}
		for all $e \in E$ it holds that $|F_z| = |F_q| = |F_{\delta}| = k$.  
	\end{lemma}
	\begin{proof}
		By \Cref{claim:cardinality} it holds that $|F_z|,|F_q| \leq k$, otherwise we violate the capacity constraint of $e_0$ or $e_k$, respectively. Moreover, $|F_{\delta}| \leq \left|\left\{\delta_i~|~i \in [k]\right\}\right| = k$. We now show that $|F_z| ,|F_q|,|F_{\delta}| = k$. Assume towards a contradiction that $|F_q|<k$, $|F_z|<k$, or $|F_{\delta}| < k$. Therefore, by the assumption and by \Cref{claim:cardinality} it holds that $|F_z|+|F_q|+|F_{\delta}| \leq 3 \cdot k-1$. Thus, 
		\begin{equation}
			\label{eq:con}
			\begin{aligned}
				w(F) ={} & w(F_z)+w(F_q)+w(F_{\delta}) \\
				={} & \sum_{i \in [k], j \in [n] \text{ s.t. } z^i_j \in F_z} w\left(z^i_j\right)+\sum_{i \in [k], j \in [n] \text{ s.t. } q^i_j \in F_q} w\left(q^i_j\right)+\sum_{i \in [k] \text{ s.t. } \delta_i \in F_{\delta}} w(\delta_i)\\
				={} & \sum_{i \in [k], j \in [n] \text{ s.t. } z^i_j \in F_z} \bigg(H+i \cdot L+ j \cdot i\bigg)+\sum_{i \in [k], j \in [n] \text{ s.t. } q^i_j \in F_q} \bigg( H+(k+1-i) \cdot L+(k+1-i) \cdot j\bigg)+\\
				{} & \sum_{i \in [k] \text{ s.t. } \delta_i \in F_{\delta}} H\\ 
				\leq{} & |F_z| \cdot \left(H+k \cdot L+ n \cdot k\right)+|F_q| \cdot \left( H+k \cdot L+k \cdot n \right)+|F_{\delta}| \cdot H\\
				\leq{} & \left( |F_z|+|F_q|+|F_{\delta}|\right) \cdot H+ \left(|F_z| +|F_q| \right) \cdot k \cdot L+\left(|F_z| +|F_q| \right) \cdot k \cdot n \\
				\leq{} & (3k-1) \cdot H+2k^2 \cdot L+2 \cdot k^2 \cdot n \\  
				={} & 3k \cdot H-H+2k^2 \cdot L+2 \cdot k^2 \cdot n \\  
				<{} & 3k \cdot H \\  
				<{} & 3k \cdot H +k \cdot (k+1) \cdot L+k \cdot (k+1) \cdot x.\\  
			\end{aligned}
		\end{equation}
		
		The third inequality holds by the assumption that $|F_z|+|F_q|+|F_{\delta}| \leq 3 \cdot k-1$ and that $\left|F_{e_0}\right| = |F_z| \leq k,|F_q|=\left|F_{e_k}\right| \leq k$ by \Cref{claim:cardinality}. The fourth inequality follows from the definition of $H$ in \eqref{eq:H}. 
		By \eqref{eq:con} we reach a contradiction that $w(F) \geq 3k \cdot H+k \cdot (k+1) \cdot L+k \cdot (k+1) \cdot x$. 
	\end{proof}

	Our next properties use the following partitions of the sets $F_z$ and $F_q$. For $i \in [k]$ define 
	\begin{equation}
		\label{eq:Fi}
		\begin{aligned}
			F^i_z ={} & \{z^i_{j} \in F_z~|~j \in [n]\}\\
			F^i_q ={} &  \{q^i_{j} \in F_z~|~j\in [n]\}.
		\end{aligned}
	\end{equation}

	In the next property, we show that if there is $i' \in [k]$ such that $|F^{i'}_z| \neq |F^{i'}_q|$, then it implies (in contradiction) that the solution $F$ has a smaller profit than its actual profit.  
	\begin{lemma}
		\label{clm:1}
		For every $i \in [k]$ it holds that $\left|F^i_z\right| = \left|F^i_q\right|$. 
	\end{lemma}
	
	\begin{proof}
		For every edge $e_h, h \in \{0,1,\ldots, k\}$ we have \begin{equation}
			\label{eq:ck}
			\sum_{i \in [k] \text{ s.t. } i>h} \left| F^i_z \right|+	\sum_{i \in [k] \text{ s.t. } i \leq h} \left| F^i_q \right|  = \left| F_{e_h} \right| \leq k.
		\end{equation} The inequality follows from \Cref{claim:cardinality}. Assume towards a contradiction that there is $i' \in [k]$ such that $|F^{i'}_z| \neq |F^{i'}_q|$. Therefore, 
		\begin{equation}
			\label{eq:n1}
			\begin{aligned}
				w(F) ={} & w(F_z)+w(F_q)+w(F_{\delta}) \\
				={} &  \bigg( |F_z|+|F_q|+|F_{\delta}|\bigg) \cdot H+\sum_{i \in [k], j \in [n] \text{ s.t. } z^i_j \in F_z} \bigg(i \cdot L+ j \cdot i\bigg)+\\
				{} & \sum_{i \in [k], j \in [n] \text{ s.t. } q^i_j \in F_q} \bigg((k+1-i) \cdot L+(k+1-i) \cdot j\bigg)\\
				\leq{} & 3k \cdot H+\sum_{i \in [k], j \in [n] \text{ s.t. } z^i_j \in F_z} \bigg(i \cdot L+ n \cdot k\bigg)+ \sum_{i \in [k], j \in [n] \text{ s.t. } q^i_j \in F_q} \bigg((k+1-i) \cdot L+k \cdot n \bigg)\\
				={} & 3k \cdot H+\sum_{i \in [k]} \left|F^i_z \right| \cdot \left(i \cdot L+ n \cdot k\right)+ \sum_{i \in [k]} \left|F^i_q \right| \cdot \bigg((k+1-i) \cdot L+k \cdot n \bigg)\\
				={} & 3k \cdot H+2 \cdot n k^2 +\sum_{i \in [k]} \left|F^i_z \right| \cdot i \cdot L+ \sum_{i \in [k]} \left|F^i_q \right| \cdot (k+1-i) \cdot L\\
			\end{aligned}
		\end{equation} By rewriting the last expression in \eqref{eq:n1} we get 
		\begin{equation}
			\label{eq:n2}
			\begin{aligned}
				w(F) \leq{} &  3k \cdot H+2 \cdot n k^2 +\sum_{h \in \{0,1,\ldots,k\}} \left(\sum_{i \in [k], i>h}  \left|F^i_z \right| \cdot L+ \sum_{i \in [k], i \leq h} \left|F^i_q \right| \cdot L\right)\\
				={} & 3k \cdot H+2 \cdot n k^2 +L \cdot \sum_{h \in \{0,1,\ldots,k\}} \left(\sum_{i \in [k], i>h}  \left|F^i_z \right|+ \sum_{i \in [k], i \leq h} \left|F^i_q \right| \right).\\
			\end{aligned}
		\end{equation} Now, if $|F^{i'}_z| > |F^{i'}_q|$ then it holds that 
		\begin{equation}
			\label{eq:ar}
			\begin{aligned}
				\sum_{i \in [k], i>i'}  \left|F^i_z \right|+\sum_{i \in [k], i \leq i'} \left|F^i_q \right|  ={} & 	\sum_{i \in [k], i>i'} \left( \left|F^i_z \right|\right)+ \left|F^{i'}_q \right|+ \sum_{i \in [k], i < i'} \left|F^i_q \right| \\
				<{} & \sum_{i \in [k], i>i'} \left( \left|F^i_z \right|\right)+ \left|F^{i'}_z \right|+ \sum_{i \in [k], i < i'} \left|F^i_q \right| \\
				={} & 	\sum_{i \in [k], i>i'-1}  \left|F^i_z \right|+\sum_{i \in [k], i \leq i'-1} \left|F^i_q \right| \\
				\leq {} & k.  
			\end{aligned}
		\end{equation} The last inequality follows from \eqref{eq:ck}. Conversely,  if $|F^{i'}_z| < |F^{i'}_q|$ we get 
		\begin{equation}
			\label{eq:ar2}
			\begin{aligned}
				\sum_{i \in [k], i>i'-1}  \left|F^i_z \right|+\sum_{i \in [k], i \leq i'-1} \left|F^i_q \right|  ={} & 	\sum_{i \in [k], i>i'} \left( \left|F^i_z \right|\right)+\left|F^{i'}_z \right|+ \sum_{i \in [k], i \leq i'-1} \left|F^i_q \right| \\
				<{} & \sum_{i \in [k], i>i'} \left( \left|F^i_z \right|\right)+ \left|F^{i'}_q \right|+ \sum_{i \in [k], i \leq i'-1} \left|F^i_q \right| \\
				={} & 	\sum_{i \in [k], i>i'}  \left|F^i_z \right|+\sum_{i \in [k], i \leq i'} \left|F^i_q \right| 
				\\
				\leq {} & k.  
			\end{aligned}
		\end{equation} The last inequality follows from \eqref{eq:ck}. By placing the bounds from \eqref{eq:ar}, \eqref{eq:ar2}, and \eqref{eq:ck} in the last expression from \eqref{eq:n2} we get 
		\begin{equation}
			\label{eq:n3}
			\begin{aligned}
				w(F) \leq{} &  3k \cdot H+2 \cdot n k^2+(k+1) \cdot k \cdot L-L\\
				<{} & 3k \cdot H+(k+1) \cdot k \cdot L\\
				<{} & 3k \cdot H +k \cdot (k+1) \cdot L+k \cdot (k+1) \cdot x.\\  
			\end{aligned}
		\end{equation}  The second inequality follows from the definition of $L$ in \eqref{eq:H}. 
		By \eqref{eq:n3} we reach a contradiction that $w(F) \geq 3k \cdot H+k \cdot (k+1) \cdot L+k \cdot (k+1) \cdot x$. 
	\end{proof}

	The next property of $F$ uses the above properties to show a stronger claim than \Cref{clm:1}: we show that 
	each of the subsets $F^i_z$ and $F^i_q$ contains a exactly one task.  
	
	\begin{lemma}
		\label{clm:2}
		For every $i \in [k]$ it holds that $\left|F^i_z\right| = \left|F^i_q\right| = 1$. 
	\end{lemma}
	
	\begin{proof}
		We start by showing a lower bound of $1$ on the discussed sets.
		\begin{claim}
			\label{clm:2H}
			For every $i \in [k]$ it holds that $\left|F^i_z\right| , \left|F^i_q\right| \geq 1$.  
		\end{claim}
		
		\begin{claimproof}
			Note that for all $j \in [n]$ and $i' \in [k] \setminus \{i\}$ it holds that {\em exactly} one task of $z^{i'}_j,q^{i'}_j$ intersects $f_i$. In addition, {\em neither} of the tasks $z^i_j$ and $q^i_j$ intersect $f_i$. 
			Then, assume towards a contradiction that $|F^i_z| =  0$; then, by \Cref{clm:1} it follows that $|F^i_q| =  0$ as well. Thus,  
			\begin{equation}
				\label{eq:fi}
				\begin{aligned}
					\left|F_{f_i}\right| ={} & |\{\delta_i\}|+\sum_{i' \in [k], i<i'} \left|F^{i'}_z\right|+\sum_{i' \in [k],i>i'} \left|F^{i'}_q\right| \\
					={} & 1+\sum_{i' \in [k], i<i'} \left(  \frac{|F^{i'}_z|+|F^{i'}_q|}{2}\right)+\sum_{i' \in [k],i>i'} \left(  \frac{|F^{i'}_z|+|F^{i'}_q|}{2}\right)\\
					={} & 1+\frac{1}{2} \cdot \sum_{i' \in [k] \setminus \{i\}} \left( \left|F^{i'}_z\right|+\left|F^{i'}_q\right| \right)\\
					=	{} & 1+ \frac{1}{2} \cdot \sum_{i' \in [k]} \left( \left|F^{i'}_z\right|+\left|F^{i'}_q\right|\right)\\
					=	{} & 1+ \frac{1}{2} \cdot \left(|F_z|+|F_q|\right)\\
					={} &	k+1. 
				\end{aligned}
			\end{equation} For the first equality, recall that $\delta_i \in F$ and that $|F^i_z| = |F^i_q| =  0$ by the assumption. The second equality follows from \Cref{clm:1}. The last equality follows since $|F_z| = |F_q| = k$ by \Cref{claim:cardinality2}.
			%
			By \eqref{eq:fi} we reach a contradiction to \Cref{claim:cardinality}. Using a symmetrical argument, if $|F^i_q| =  0$ then by \Cref{clm:1} it follows that $|F^i_z| =  0$; thus, using \eqref{eq:fi} we reach a contradiction to \Cref{claim:cardinality} in this case as well. Therefore, $|F^i_z|,|F^i_q|\geq 1$. 
		\end{claimproof}
		By \Cref{clm:2H}, we only need to show that  $\left|F^i_z\right| , \left|F^i_q\right| \leq 1$.  Assume towards a contradiction that $|F^i_z| \geq 2$. 
		Therefore, since $|F_z|+|F_q| = 2 \cdot k$ by \Cref{claim:cardinality2} and that $|F^i_z| \geq 2$, then
		there is $i' \in [k]$ such that $|F^{i'}_z| = 0$ or $|F^{i'}_q| = 0$; this is a contradiction to \Cref{clm:2H}. The complementary case where $|F^i_q| \geq 2$ is analogous to the above.  \end{proof}

	After establishing the general structure of the solution $F$, we show the more fine grained structure of the solution, w.r.t. to the parameters $j$ of the tasks. For all $i \in [k]$ and $j \in [n]$ such that $z^i_j \in F^i_z$, let $j_{z}(i) = j$; in addition, for $j' \in [n]$ such that $q^i_{j'} \in F^i_q$ let $j_{q}(i) = j'$. Also, let $j(z^i_j) = j$ and $j(q^i_{j'}) = j'$. By \Cref{clm:2}, for every $i \in [k]$ it holds that $\left|F^i_z\right| = \left|F^i_q\right| = 1$. Thus, $j_{z}(i)$ and $j_{q}(i)$ are well defined. Even though $F$ has freedom in the choice of $j_z(i)$ for $i \in [k]$, in the next result we show that it must be the same choice also for $j_q(i)$. Moreover, we also show that for every index $h \in \{0,1,\ldots,k\}$ along the path, the sum of the parameters  $j_z(i)$, $j_q(i)$ that intersect $e_h$ is exactly $x$.  

	\begin{lemma}
		\label{claim:3}
		For all $i \in [k]$ it holds that $j_z(i) = j_q(i)$. Moreover, for all $h \in \{0,1,\ldots,k\}$: $$	\sum_{i \in [k], i>h}  j_z(i)+\sum_{i \in [k], i \leq h} j_q(i) = x.$$
	\end{lemma}

	\begin{proof} We first show that the sum of parameters $j$ over tasks intersecting any edge $e_i$ is at most $x$. 
		\begin{claim}
			\label{claim:3H}
			For all $e_i \in E, i \in \{0,1,\ldots,k\}$ it holds that $$\sum_{t \in F_z \cup F_q  \text{ s.t. } e_i \in P(t)~} j(t) \leq x$$
		\end{claim}
		\begin{claimproof}
			Assume towards a contradiction that there is $i \in \{0,1,\ldots,k\}$ such that \begin{equation}
				\label{eq:Acon}
				\sum_{t \in F_z \cup F_q  \text{ s.t. } e_i \in P(t)~} j(t) > x.
			\end{equation} Then, 
			\begin{equation}
				\label{eq:c'}
				\begin{aligned}
					\sum_{t \in F, e_i \in P(t)} d(t) ={} &  	\sum_{t \in F_z, e_i \in P(t)} d(t)+\sum_{t \in F_q, e_i \in P(t)} d(t)\\
					={} & \sum_{i' \in [k] \text{ s.t. } i'>i} d\left(z^i_{j_z(i')}\right)+\sum_{i' \in [k] \text{ s.t. } i'\leq i} d\left(z^i_{j_q\left(i'\right)}\right) \\
					={} & \sum_{i' \in [k] \text{ s.t. } i'>i}  \left( Q+ j_z(i') \cdot W+a^{i'}_{j_z(i')} \right)+ \sum_{i' \in [k] \text{ s.t. } i'\leq i}  \left( Q+ j_q\left(i'\right) \cdot W+\frac{2 \cdot B}{k}-a^{i'}_{j_q\left(i'\right)} \right) \\
					\geq{} &  k \cdot Q+ W \cdot \left(\sum_{i' \in [k] \text{ s.t. } i'>i} j_z(i')+ \sum_{i' \in [k] \text{ s.t. } i'\leq i} j_q\left(i'\right) \right)\\
					={} &  k \cdot Q+W \cdot \sum_{t \in F_z \cup F_q  \text{ s.t. } e_i \in P(t)~} j(t) \\
					\geq{} &  k \cdot Q+W \cdot (x+1)\\
					={} & k \cdot Q+x \cdot W+W\\
					>{} & k \cdot Q+x \cdot W+r(S)\\
					\geq {} & u(e_i). 
				\end{aligned}
			\end{equation} 
			The second inequality holds by \eqref{eq:Acon}. The third inequality follows from \eqref{eq:W}. By \eqref{eq:c'} we reach a contradiction to the feasibility of $F$.  
		\end{claimproof}

		Assume towards a contradiction that the statement of the lemma does not hold. Thus, at least one of the following conditions hold.
		\begin{itemize}
			\item 	 (a) There is $i' \in [k]$ such that $j_z(i') \neq j_q\left(i'\right)$.
			\item (b) There is $h \in \{0,1,\ldots,k\}$ such that $$	\sum_{i \in [k], i>h}  j_z(i)+\sum_{i \in [k], i \leq h} j_q(i) < x.$$ 
		\end{itemize} 
		If one of the conditions (a) or (b) holds, we reach a contradiction to the profit guarantee of $F$. Observe that for every edge $e_h$ where  $h \in \{0,1,\ldots, k\}$ it holds that  \begin{equation}
			\label{eq:ck'}
			\sum_{i \in [k] \text{ s.t. } i>h} j_z(i)+	\sum_{i \in [k] \text{ s.t. } i \leq h} j_q(i)=  \sum_{t \in F_z \cup F_q  \text{ s.t. } e_h \in P(t)~} j(t) \leq x
		\end{equation} The inequality follows from \Cref{claim:3H}. 
		Therefore, 
		\begin{equation}
			\label{eq:n1'}
			\begin{aligned}
				w(F) ={} & w(F_z)+w(F_q)+w(F_{\delta}) \\
				={} &  \bigg( |F_z|+|F_q|+|F_{\delta}|\bigg) \cdot H+\sum_{i \in [k]} \bigg( L \cdot \left( \left|F^i_z \right| \cdot i+ \left|F^i_q \right| \cdot (k+1-i) \right)\bigg)+ \\
				{}& \sum_{i \in [k]} i \cdot j_z(i)+\sum_{i \in [k]} (k+1-i) \cdot j_q(i)\\
				={} & 3k \cdot H+ (k+1) \cdot k \cdot L+\sum_{i \in [k]} \bigg( i \cdot j_z(i)+(k+1-i) \cdot j_q(i) \bigg)\\
			\end{aligned}
		\end{equation}
		By rewriting the last expression in \eqref{eq:n1'} we get 
		\begin{equation}
			\label{eq:n2'}
			\begin{aligned}
				w(F) ={} &  3k \cdot H+ (k+1) \cdot k \cdot L+\sum_{h \in \{0,1,\ldots,k\}} \left(\sum_{i \in [k], i>h}  j_z(i)+ \sum_{i \in [k], i \leq h} j_q(i)\right)\\
			\end{aligned}
		\end{equation} Now, if assumption (a) holds and $j_z(i') > j_q\left(i'\right)$ then it holds that 
		\begin{equation}
			\label{eq:ar'}
			\begin{aligned}
				\sum_{i \in [k], i>i'}  j_z(i)+ \sum_{i \in [k], i \leq i'} j_q(i) ={} & 		\sum_{i \in [k], i>i'}  \big( j_z(i)\big)+ j_q\left(i'\right)+ \sum_{i \in [k], i < i'} j_q(i) \\
				<{} &	\sum_{i \in [k], i>i'}  \big( j_z(i)\big) + j_z(i') +\sum_{i \in [k], i < i'} j_q(i)\\
				={} & 	\sum_{i \in [k], i>i'-1}  j_z(i)+\sum_{i \in [k], i \leq i'-1} j_q(i) \\
				\leq {} & x.  
			\end{aligned}
		\end{equation} The last inequality follows from \eqref{eq:ck'}. Conversely,  if assumption (a) holds and $j_z(i') < j_q\left(i'\right)$  we get 
		\begin{equation}
			\label{eq:ar2'}
			\begin{aligned}
				\sum_{i \in [k], i>i'-1}  j_z(i)+\sum_{i \in [k], i \leq i'-1} j_q(i)  ={} & 	\sum_{i \in [k], i>i'} \left( j_z(i) \right)+j_z(i')+ \sum_{i \in [k], i \leq i'-1} j_q(i) \\
				<{} & \sum_{i \in [k], i>i'} \left(j_z(i)\right)+ j_q\left(i'\right)+ \sum_{i \in [k], i \leq i'-1} j_q(i) \\
				={} & 	\sum_{i \in [k], i>i'}  j_z(i)+\sum_{i \in [k], i \leq i'} j_q(i) 
				\\
				\leq {} & x.  
			\end{aligned}
		\end{equation} The last inequality follows from \eqref{eq:ck'}. Clearly, the same bound as in \eqref{eq:ar'} and \eqref{eq:ar2'} trivially holds if assumption (b) is true, i.e., \begin{equation}
			\label{eq:n2'Sec}
			\sum_{i \in [k], i>h}  j_z(i)+\sum_{i \in [k], i \leq h} j_q(i) < x.
		\end{equation} By placing the bounds from \eqref{eq:ck'},  \eqref{eq:ar'}, \eqref{eq:ar2'}, and \eqref{eq:n2'Sec} in the last expression from \eqref{eq:n2'}, we get 
		\begin{equation}
			\label{eq:n3'}
			\begin{aligned}
				w(F) \leq{} &  3k \cdot H+ (k+1) \cdot k \cdot L+(k+1) \cdot k \cdot x-1 < 3k \cdot H +k \cdot (k+1) \cdot L+k \cdot (k+1) \cdot x 
			\end{aligned}
		\end{equation}  
		By \eqref{eq:n3'} we reach a contradiction since $w(F) \geq 3k \cdot H+k \cdot (k+1) \cdot L+k \cdot (k+1) \cdot x$.  
	\end{proof}

	Using the above properties, we can prove \Cref{lem:only}. 
	\subsection{Proof of \Cref{lem:only}}
	
	Let $k \leq x \leq k \cdot n$ and let $F$ be a solution for $U_{x}$ of profit at least $$3k \cdot H+k \cdot (k+1) \cdot L+k \cdot (k+1) \cdot x.$$
	As we do not have further assumptions on the solution, $F$ satisfies the conditions of \Cref{clm:1,clm:2,claim:3} described above. Thus, 
	$F_z \cup F_q$ satisfy that 
	\begin{equation}
		\label{eq:FF}
		\begin{aligned}
			F_z = \left\{z^i_{j_z(i)}~|~ i \in [k]\right\} =  \left\{z^i_{j_q(i)}~|~ i \in [k]\right\}\\
			F_q = \left\{q^i_{j_{z}(i)}~|~ i \in [k]\right\} = \left\{q^i_{j_{q}(i)}~|~ i \in [k]\right\}\\
		\end{aligned}
	\end{equation}
	Define a solution for the SSM instance $S$ as: $r(i) = j_z(i) = j_q(i)$ for all $i \in [k]$. To prove \Cref{lem:only}, we show that $r(1),\ldots,r(k)$ is indeed a solution for $S$. We use the following property of the above solution. 
	
	\begin{claim}
		\label{claim:sum_r(i)}
		$\sum_{i \in [k]} r(i) = x$.
	\end{claim}
	\begin{claimproof}
		Let $h = 0$. then, by \Cref{claim:3} it holds that	$$\sum_{i \in [k]} r(i) = \sum_{i \in [k], i > h} j_z(i)  = \sum_{i \in [k], i>h}  j_z(i)+\sum_{i \in [k], i \leq h} j_q(i) =x.$$
	\end{claimproof}
	By the feasibility of the solution $F$ for the edge $e_0$ it holds that 
	\begin{equation}
		\label{eq:e0F}
		\sum_{t \in F, e_0 \in P(t)} d(t) \leq u(e_0) = k \cdot Q+x \cdot W+B. 
	\end{equation}
	The above sum can be expressed as: 
	\begin{equation}
		\label{eq:Ls'}
		\begin{aligned}
			\sum_{t \in F, e_0 \in P(t)} d(t) ={} &  	\sum_{t \in F_z, e_0 \in P(t)} d(t)\\
			={} & \sum_{i \in [k]} d\left(z^i_{r(i)}\right) \\
			={} & \sum_{i \in [k]}  \left( Q+ r(i) \cdot W+a^i_{r(i)} \right)\\
			={} &  |F_z| \cdot Q+\sum_{i \in [k]} \left(r(i) \cdot W+a^i_{r(i)}\right)\\
			={} &  k \cdot Q+W \cdot x+\sum_{i \in [k]} a^i_{r(i)}. 
		\end{aligned}
	\end{equation} 
	The last equality follows by \Cref{claim:cardinality2}, \Cref{claim:3}, and \Cref{claim:sum_r(i)}. 
	By \eqref{eq:e0F} and \eqref{eq:Ls'} it holds that \begin{equation}
		\label{eq:DC'}
		\sum_{i \in [k]} a^i_{r(i)} \leq B. 
	\end{equation} We now show that $\sum_{i \in [k]} a^i_{r(i)} \geq B$. 	By the feasibility of the solution $F$ for the edge $e_k$ it holds that 
	\begin{equation}
		\label{eq:ekF}
		\sum_{t \in F, e_k \in P(t)} d(t) \leq u(e_k) = k \cdot Q+x \cdot W+B. 
	\end{equation} Rewriting the first expression alternatively, we have
	\begin{equation}
		\label{eq:Ls2'}
		\begin{aligned}
			\sum_{t \in F, e_k \in P(t)} d(t) ={} &  	\sum_{t \in F_q, e_k \in P(t)} d(t)\\
			={} & \sum_{i \in [k]} d\left(q^i_{r(i)}\right) \\
			={} & \sum_{i \in [k]}  \left( Q+ r(i) \cdot W+\frac{2 \cdot B}{k} - a^i_{r(i)} \right)\\
			={} &  |F_q| \cdot Q+\sum_{i \in [k]} \left(r(i) \cdot W+\frac{2 \cdot B}{k} - a^i_{r(i)}\right)\\
			={} &  k \cdot Q+W \cdot x+2 \cdot B - \sum_{i \in [k]} a^i_{r(i)}.
		\end{aligned}
	\end{equation} 
	The last equality follows by \Cref{claim:cardinality2}, \Cref{claim:3}, and \Cref{claim:sum_r(i)}. 
	By \eqref{eq:ekF} and \eqref{eq:Ls2'} it holds that \begin{equation}
		\label{eq:DC2'}
		\sum_{i \in [k]} a^i_{r(i)} \geq B. 
	\end{equation} By \eqref{eq:DC'} and \eqref{eq:DC2'} it holds that $\sum_{i \in [k]} a^i_{r(i)} = B$. Hence, 	$r(1),\ldots,r(k)$ is a solution for $S$. \qed 

	By the above, the proof of \Cref{thm:aux} follows.  
 

\noindent {\bf Proof of \Cref{thm:aux}:} The proof follows from \Cref{lem:if} and \Cref{lem:only}. \qed

			\section{Hardness of SSM}
			\label{sec:SSM}
			In this section, we prove \Cref{lem:SSM}. Our hardness result is based on the classic $k$-subset sum problem, known to be W[1]-Hard. For completeness, we define $k$-subset sum below. We use a more technical definition (yet equivalent) to simplify the proofs. 

			\begin{definition}
				\label{def:SS}
				{\bf $k$-Subset Sum (k-SS):}  For some $n \in \mathbb{N}$ and $k \in [n]$, let $A = \{a_1,\ldots,a_n\} \subset \mathbb{R}_{> 0}$ be a set of $n$ numbers, and let $B \in \mathbb{R}_{> 0}$ be a target value. 
				A collection of numbers $x_1,\ldots,x_n \in \{0,1\}$ 
				is called a {\em solution} for the instance if $\sum_{i \in [k]} x_i \cdot a_{i} = B$ and $\sum_{i \in [n]} x_i = k$.  The goal is to decide if there is a solution. The {\em parameter} of the instance is $k$.  
			\end{definition}

			As an intermediate step towards the proof of \Cref{lem:SSM}, we first prove the hardness of the following variant of k-SS in which each number can be chosen more than once, to a total of $k$ selections of numbers overall. 
			
			\begin{definition}
				\label{def:SSR}
				{\bf Subset Sum with Repetitions (SSR):}  For some $n \in \mathbb{N}$ and $k \in [n]$, let $A = \{a_1,\ldots,a_n\} \subset \mathbb{R}_{> 0}$ be a set of $n$ numbers, and let $B \in \mathbb{R}_{> 0}$ be a target value. 
				A {\em solution} for the instance is $x_1,\ldots,x_n \in \{0,1,\ldots,k\}$ such that $\sum_{i \in [n]} x_i \cdot a_i = B$ and $\sum_{i \in [n]} x_i = k$. The goal is to decide if there is a solution. The {\em parameter} of the instance is $k$. 
			\end{definition}
			
			For constructing a reduction from k-SS to SSR, we rely on the following auxiliary result. Given a number $k \in \mathbb{N}_{>0}$ and $i \in [k]$, define $L^k_i = i \cdot k^{i-1}$; moreover, define $L^k = \sum_{i \in [k]} L^k_i$.  
			
			\begin{lemma}
				\label{lem:Lk}
				For all $k \in \mathbb{N}_{>0}$ and coefficients $x_1,\ldots,x_k \in \{0,1,\ldots,k\}$ such that $\sum_{i \in [n]} x_i = k$, it holds that $\sum_{i \in [k]} x_i \cdot L^k_i = L^k$ if and only if $x_i = 1$ for all $i \in [k]$. 
			\end{lemma}
			
			\begin{proof}
				Fix $k \in \mathbb{N}_{>0}$. For all $t \in [k]$ let $L(t) = \sum_{i \in [t]} L^k_i$. We prove by induction on $t = 1,2,\ldots,k$ that for every coefficients $x_1,\ldots,x_t \in \{0,1,\ldots,k\}$ such that $\sum_{i \in [n]} x_i = t$ it holds that $\sum_{i \in [t]} x_i \cdot L^k_i = L(t)$ if and only if $x_i = 1$ for all $i \in [t]$.   For the base case, let $t = 1$. Then, $L(1) = L^k_1 = 1$ which satisfies $x_1 \cdot 1 = L(1)$ if and only if $x_1 = 1$. Assume that for some $t \in [k-1]$ it holds that: for every coefficients $x_1,\ldots,x_t \in \{0,1,\ldots,k\}$ such that $\sum_{i \in [n]} x_i = t$ it holds that $\sum_{i \in [t]} x_i \cdot L^k_i = L(t)$ if and only if $x_i = 1$ for all $i \in [t]$. For the step of the induction, consider some coefficients  $x_1,\ldots,x_{t+1} \in \{0,1,\ldots,k\}$ such that $\sum_{i \in [t+1]} x_i = t+1$. For the first direction, assume that $\sum_{i \in [t+1]} x_i \cdot L^k_i = L(t+1)$ and we will prove that $x_i = 1$ for all $i \in [t+1]$. Assume towards a contradiction that $x_{t+1} \neq 1$. We consider two cases. 
				\begin{itemize}
					\item 	If $x_{t+1} = 0$, then since $L^k_t \geq L^k_{i}$ for all $i \in [t]$ it follows that $$\sum_{i \in [t+1]} x_i \cdot L^k_i \leq \sum_{i \in [t+1]} x_i \cdot L^k_t = (t+1) \cdot L^k_t \leq k \cdot L^k_t = k \cdot t \cdot k^{t-1}  < (t+1) \cdot k^{t} = L(t+1).$$
					\item Conversely, $x_{t+1} \geq 2$. Then,  
					$$\sum_{i \in [t+1]} x_i \cdot L^k_i \geq 2 \cdot L^k_{t+1} = L^k_{t+1} +(t+1) \cdot k^t  >  L^k_{t+1} +t \cdot L^k_t \geq \sum_{i \in [t+1]} L^k_i = L(t+1).$$
				\end{itemize} In both cases above we reach a contradiction since $\sum_{i \in [t+1]} x_i \cdot L^k_i = L(t+1)$. Therefore, $x_{t+1} = 1$; hence, since $\sum_{i \in [t+1]} x_i \cdot L^k_i = L(t+1)$ and $x_{t+1} = 1$ it implies that  $$\sum_{i \in [t]} x_i \cdot L^k_i = \sum_{i \in [t+1]} x_i \cdot L^k_i - L^k_{t+1} = L(t+1)- L^k_{t+1} = L(t).$$ Thus, by  the assumption of the induction it follows that $x_i = 1$ for all $i \in [t+1]$. We now prove the second direction (of the inductive step). Let $x_1,\ldots,x_{t+1} \in \{0,1,\ldots,k\}$ such that $\sum_{i \in [t+1]} x_i = t+1$ and $x_i = 1$ for all $i \in [t+1]$. Then, 
				$$\sum_{i \in [t+1]} x_i \cdot L^k_i  = \sum_{i \in [t+1]} L^k_i = L(t+1)$$  by definition.  By the above, the proof follows. 
			\end{proof}
			
			We now define a reduction from k-SS to SSR. Let $A = \{a_1,\ldots,a_n\} \subset \mathbb{R}_{> 0}$, $k \in \mathbb{N}_{>}$, and $B \in \mathbb{R}_{> 0}$ be a k-SS instance $I = (A,k,B)$. The main idea is to use {\em color coding}, where we partition the set of numbers $A$ into a small (that is, $k$) number of classes. For each class $i$ we give a different value $w(i)$ as a scaling factor. As color coding has a deterministic algorithm, we can create in polynomial deterministic time a polynomial number of partitions of $A$ to classes in $n$, parametrized by $k$, such that there exists some partition to classes that assign numbers of a given solution for $I$ to different classes (or, gives them a different scaling factor). The following lemma follows from an immediate interpretation of the color coding scheme to our notation, as well as the results of  \cite{naor1995splitters}, that derandomizes the color coding scheme.    
			
			\begin{lemma}
				\label{lem:colorCoding}
				There is a computable function $f:\mathbb{N} \rightarrow \mathbb{N}$ and an algorithm \textnormal{\textsf{Color-Coding}} that given a $k$-\textnormal{SS} instance $I = (A,k,B)$, where $A = \{a_1,\ldots,a_n\}$, returns in time $f(k) \cdot |I|^{O(1)}$ functions $w_1,\ldots, w_J$ such that the following holds. 
				\begin{enumerate}
					\item $J = f(k) \cdot |I|^{O(1)}$.
					\item For all $j \in [J]$ it holds that $w_j$ is a function $w_j:[n] \rightarrow [k]$ 
					\item For all $S \subseteq [n]$ such that $|S| = k$, there is $j \in [J]$ such that for every $x,y \in S$ where $x \neq y$ it holds that $w_j(x) \neq w_j(y)$. 
				\end{enumerate}
			\end{lemma}
			
			Using \Cref{lem:colorCoding}, we show the hardness of SSR. 

				\begin{lemma}
					\label{lem:SSR-Hard}
					Unless \textnormal{W[1] = FPT}, for every function $g:\mathbb{N} \rightarrow \mathbb{N}$ there is no algorithm that decides \textnormal{SSR} in time $g(k) \cdot n^{O(1)}$, where $n$ is the encoding size of the instance and $k$ is the parameter. 
				\end{lemma}
				
				\begin{proof}
					Assume that \textnormal{W[1] $\neq$ FPT} and assume towards a contradiction that there is a function $g:\mathbb{N} \rightarrow \mathbb{N}$ and an algorithm $\cA$ decides \textnormal{SSR} in time $g(k) \cdot n^{O(1)}$, where $n$ is the encoding size of the instance and $k$ is the parameter. We show that we can decide k-SS using the following algorithm $\cB$ relying on the existence of $\cA$.  Let $A = \{a_1,\ldots,a_n\} \subset \mathbb{R}_{> 0}$, $k \in \mathbb{N}_{>}$, and $B \in \mathbb{R}_{> 0}$ be a k-SS instance $I = (A,k,B)$ and let $\textsf{scale}(A) = 2 k \cdot \sum_{a \in A} a$ be the total sum of the numbers in $A$ scaled by a factor of $2k$. Algorithm $\cB$ goes as follows on input $I$.
					\begin{enumerate}
						\item Execute Algorithm \textsf{Color-Coding} on $I$; let $w_1,\ldots, w_J$ be the output. 
						\item For all $j \in [J]$:
						\begin{enumerate}
							\item  Define the {\em reduced} SSR instance of $I$ and $j$ as $A^j = \{a^j_1,\ldots,a^j_n\} \subset \mathbb{R}_{> 0}$, $k$, and $B' \in \mathbb{R}_{> 0}$ such that the following holds. 
							\begin{enumerate}
								\item For all $i \in [n]$ define $a^j_i = L^k_{w_j(i)} \cdot \textsf{scale}(A)+a_i$.  
								\item  Define $B' = L^k \cdot \textsf{scale}(A)+B$. 
							\end{enumerate}
							
							\item 		Execute $\cA$ on the SSR instance $R^j_I = (A^j,k,B')$.
							\item If $\cA$ returns that $R^j_I$ has a solution: Return that $I$ has a solution. 
						\end{enumerate}
						\item Return that $I$ does not have a solution. 
					\end{enumerate}
					Let $f:\mathbb{N} \rightarrow \mathbb{N}$ such that the running time of \textsf{Color-Coding} is bounded by $f(k) \cdot |I|^{O(1)}$ and $J = f(k) \cdot |I|^{O(1)}$; there is such a function $f$ by \Cref{lem:colorCoding}. Clearly, for all $j \in J$ the encoding size of the reduced instance $R^j_I$ is  bounded by $|I|^{O(1)}$ (recall that $|I|$ is the encoding size of $I$). Thus, $\cB$ runs in time $\left(f(k)+g(k) \right) \cdot |I|^{O(1)}$ by \Cref{lem:colorCoding} and by the running time guarantee of $\cA$.
					\begin{claim}
						\label{claim:solIFF}
						$\cB$ returns that $I$ has a solution if and only if there is $j \in J$ such that $R^j_I$ has a solution
					\end{claim}
					\begin{claimproof}
						First, assume that $I$ has a solution. Let $x_1,\ldots,x_n \in \{0,1\}$ be a solution for $I$; that is, $\sum_{i \in [n]} x_i \cdot a_{i} = B$ and $\sum_{i \in [n]} x_i = k$. Let $S = \{i \in [n]~|~x_i = 1\}$; note that $|S| = k$. Thus,  by \Cref{lem:colorCoding}, there is $j \in [J]$ such that: for every $x,y \in S$ where $x \neq y$ it holds that $w_j(x) \neq w_j(y)$. We show that $x_1,\ldots,x_n$ is a solution for $R^j_I$ as well. Observe that
						\begin{equation*}
							\label{eq:Seq}
							\sum_{i \in [n]} x_i \cdot a^j_{i} = \sum_{i \in [n]} x_i \cdot \left(  L^k_{w_j(i)} \cdot \textsf{scale}(A)+a_i \right) = \textsf{scale}(A) \cdot \sum_{i \in [k]} L^k_{i} + \sum_{i \in [n]} x_i \cdot a_i = L^k \cdot \textsf{scale}(A)+B. 
						\end{equation*} The second equality follows from \Cref{lem:colorCoding}, since for all $t \in [k]$ there is exactly one $i \in [n]$ such that $ L^k_{w_j(i)} = L^k_t$ (note that $i \in [n]$ such that $x_i = 0$ does not change the sum). The last equality holds since $x_1,\ldots,x_n$ is a solution for $I$. Hence, $x_1,\ldots,x_n$ is a solution for $R^j_I$ as well. Therefore, $\cA$ returns that $R^j_I$ has a solution; as a result, $\cB$ returns that $I$ has a solution.
						
						Conversely, assume that $\cB$ returns that $I$ has a solution. Therefore, there is $j \in [J]$ such that $\cA$ returns that $R^j_I$ has a solution. Since $\cA$ is assumed to decide SSR correctly, it follows that $R^j_I$ has a solution. Thus, by \Cref{def:SSR} there are numbers $x_1,\ldots,x_n \in \{0,1,\ldots,k\}$ such that $\sum_{i \in [n]} x_i \cdot a^j_i = B'$ and $\sum_{i \in [n]} x_i = k$. Let $\alpha =  \sum_{i \in [n]} x_i \cdot  L^k_{w_j(i)} \cdot \textsf{scale}(A)$. Assume towards a contradiction that $\alpha \neq L^k \cdot \textsf{scale}(A)$; since $\alpha \leq \sum_{i \in [n]} x_i \cdot a^j_{i}$ (by the definition of the numbers $a^j_i$) and $\sum_{i \in [n]} x_i \cdot a^j_{i} = L^k \cdot \textsf{scale}(A)+B$ (as $x_1,\ldots,x_n$ is a solution for $R^j_I$), therefore $\alpha < L^k \cdot \textsf{scale}(A)$. Moreover, because $\alpha$ is a multiple of $\textsf{scale}(A)$ it follows that $\alpha \leq \left( L^k-1\right) \cdot \textsf{scale}(A)$. Then, 
						\begin{equation*}
							\label{eq:Siii}
							\begin{aligned}
								\sum_{i \in [n]} x_i \cdot a^j_{i} = \alpha +\sum_{i \in [n]} a_i \leq \left( L^k-1\right) \cdot \textsf{scale}(A)+\sum_{i \in [n]} a_i \leq L^k \cdot \textsf{scale}(A) < L^k \cdot \textsf{scale}(A)+B = B'.  
							\end{aligned}
						\end{equation*} Since $x_1,\ldots,x_n$ is a solution for $R^j_I$, we reach a contradiction; we conclude that $$\alpha =   \sum_{i \in [n]} x_i \cdot  L^k_{w_j(i)} \cdot \textsf{scale}(A) = L^k \cdot \textsf{scale}(A).$$ 
						Consequently, $\sum_{i \in [n]} x_i \cdot  L^k_{w_j(i)} = L^k$.  Thus, by \Cref{lem:Lk}, 
      for every $t\in [k]$ it holds that $$\sum_{i\in [n] \text{ s.t. } w_j(i)=t} x_i = 1.$$ This implies that $x_i\in\{0,1\}$ for every $i\in[n]$. 
      Since $x_1,\ldots,x_n$ is a solution for $R^j_I$,
      
						\begin{equation}
							\label{eq:sum_xi}
							\alpha +\sum_{i \in [n]} x_i\cdot a_i  = \sum_{i \in [n]} x_i \cdot a^j_{i}  = L^k \cdot \textsf{scale}(A)+B = B'.  
						\end{equation} Since $\alpha = L^k \cdot \textsf{scale}(A)$, by \eqref{eq:sum_xi} it holds that $\sum_{i \in [n]} x_i\cdot a_i = B$. 
      By the above and because $\sum_{i \in [n]} x_i = k$, it holds that $x_1,\ldots, x_n$ is a solution for $I$ as well, and in particular, $I$ has a solution. The proof follows. 
					\end{claimproof}

					By \Cref{claim:solIFF} it holds that $\cB$ returns that $I$ has a solution if and only if $I$ indeed has a solution; that is, $\cB$ correctly decides k-SS. 	 Since k-SS is known to be W[1]-Hard \cite{downey1995fixed}, we reach a contradiction and the statement of the lemma follows.  
				\end{proof}

				We can finally reduce SSR to SSM. 
				
				\noindent {\bf Proof of \Cref{lem:SSM}:} 
				
				Assume that \textnormal{W[1] $\neq$ FPT} and assume towards a contradiction that there is a function $f:\mathbb{N} \rightarrow \mathbb{N}$ and an algorithm $\cA$ decides \textnormal{SSM} in time $f(k) \cdot n^{O(1)}$, where $n$ is the encoding size of the instance and $k$ is the parameter. We show that we can decide SSR using the following algorithm $\cB$ relying on the existence of $\cA$.  Let $A = \{a_1,\ldots,a_n\} \subset \mathbb{R}_{> 0}$, $k \in \mathbb{N}_{>}$, and $B \in \mathbb{R}_{> 0}$ be an SSR instance $I = (A,k,B)$ and let $\textsf{scale}(A) = 2 k \cdot \sum_{a \in A} a$. Algorithm $\cB$ goes as follows on input $I$.
				\begin{enumerate}
					\item  Define the {\em reduced} SSM instance $R_I$ of $I$ as $A_1 = \{a^1_1,\ldots,a^1_n\}, A_2 = \{a^2_1,\ldots,a^2_n\}, \ldots, A_k = \{a^k_1,\ldots,a^k_n\} \subset \mathbb{R}_{> 0}$ and $B' \in \mathbb{R}_{> 0}$ such that the following holds. 
					\begin{enumerate}
						\item For all $j \in [k]$ and $i \in [n]$ define $a^j_i = L^k_{j} \cdot \textsf{scale}(A)+a_i$.  
						\item  Define $B' = L^k \cdot \textsf{scale}(A)+B$. 
					\end{enumerate}
					
					\item 		Execute $\cA$ on the reduced SSM instance $R_I$.  
					\item Return that $I$ has a solution if and only if $\cA$ returns that $R_I$ has a solution. 
				\end{enumerate}
				
				Since the encoding size of the reduced instance $R_I$ is  bounded by $|I|^{O(1)}$, $\cB$ runs in time $f(k) \cdot |I|^{O(1)}$ by the running time guarantee of $\cA$. We show below that $\cB$ correctly decides $I$. \begin{itemize}
					\item If $I$ has a solution. Then, there are $x_1,\ldots, x_n \in \{0,1,\ldots,k\}$ such that $\sum_{i \in [n]} x_i = k$ and $\sum_{i \in [n]} x_i \cdot a_i = B$. Let $r:[k] \rightarrow [n]$ be a function such that for all $i \in [n]$ it holds that $$x_i = \left|\{j \in [k]~|~r(j) = i\}\right|.$$ 
					Since  $\sum_{i \in [n]} x_i = k$ there is such a function $r$. We show that $r(1),\ldots, r(k)$ is a solution for $R_I$. 
					\begin{equation*}
						\label{eq:r(i)}
						\begin{aligned}
							\sum_{j \in [k]} a^j_{r(j)} ={} &  \sum_{j \in [k]} \left(  L^k_{j} \cdot \textsf{scale}(A)+a_{r(j)} \right) \\
							={} & \textsf{scale}(A) \cdot \sum_{j \in [k]} L^k_{j} + \sum_{j \in [k]} a_{r(j)} \\
							={} & L^k \cdot \textsf{scale}(A)+ \sum_{j \in [k]} a_{r(j)} \\ 
							={} & L^k \cdot \textsf{scale}(A)+ \sum_{i \in [n]}  \left|\{j \in [k]~|~r(j) = i\}\right| \cdot a_i \\ 
							={} & L^k \cdot \textsf{scale}(A)+ \sum_{i \in [n]}  x_i \cdot a_i \\ 
							={} & L^k \cdot \textsf{scale}(A)+B\\
							={} & B'. 
						\end{aligned}
					\end{equation*}
					The third equality follows from \Cref{lem:Lk}. The second equality from the end holds since $x_1,\ldots,x_n$ is a solution for $I$. Hence, $r(1),\ldots, r(k)$ is a solution for $R_I$ as well. Therefore, $\cA$ returns that $R_I$ has a solution; as a result, $\cB$ returns that $I$ has a solution.
					
					\item If $\cB$ returns that $I$ has a solution. Then, $\cA$ returns that $R_I$ has a solution. As $\cA$ correctly decides SSM, it follows that $R_I$ has a solution; thus, there are $r(1),\ldots, r(k) \in [n]$ satisfying $\sum_{j \in [k]} a^j_{r(j)} = B'$. Define 
					$$x_i = \left|\{j \in [k]~|~r(j) = i\}\right|.$$ 
					We show that $x_1,\ldots, x_n$ is a solution for $I$. First, note that $x_i \in \{0,1,\ldots,k\}$ for all $i \in [n]$ and $$\sum_{i \in [n]} x_i = \sum_{i \in [n]} \left|\{j \in [k]~|~r(j) = i\}\right| = k.$$
					Finally, 
					\begin{equation*}
						\label{eq:sB}
						\begin{aligned}
							\sum_{i \in [n]}  x_i \cdot a_i ={} & \sum_{i \in [n]}  \left|\{j \in [k]~|~r(j) = i\}\right| \cdot a_i \\
							={} &  \sum_{j \in [k]} a_{r(j)}\\
							={} &  - L^k \cdot \textsf{scale}(A)+ L^k \cdot \textsf{scale}(A) + \sum_{j \in [k]} a_{r(j)}\\
							={} &  - L^k \cdot \textsf{scale}(A) + \textsf{scale}(A) \cdot \sum_{j \in [k]} L^k_{j} + \sum_{j \in [k]} a_{r(j)} \\
							={} &  - L^k \cdot \textsf{scale}(A) +  \sum_{j \in [k]} \left(  L^k_{j} \cdot \textsf{scale}(A)+a_{r(j)} \right) \\
							={} &  - L^k \cdot \textsf{scale}(A) + 	\sum_{j \in [k]} a^j_{r(j)} \\
							={} &  - L^k \cdot \textsf{scale}(A) + 	B'\\
							={} &  - L^k \cdot \textsf{scale}(A) + 	L^k \cdot \textsf{scale}(A)+B\\
							={} & B.
						\end{aligned}
					\end{equation*}  The fourth equality follows from \Cref{lem:Lk}. The third equality from the end holds since $r(1),\ldots, r(k)$ is a solution for $R_I$. Thus, $x_1,\ldots, x_n$ is a solution for $I$. 
				\end{itemize}
				We conclude that $\cB$ correctly decides $I$ in time $f(k) \cdot |I|^{O(1)}$. This is a contradiction to \Cref{lem:SSR-Hard}. \qed

    \newpage


\bibliographystyle{splncs04}
\bibliography{bibfile}

\end{document}